\documentclass[12pt]{article}
 
 \usepackage{pdfpages}

\usepackage[margin=0.9in]{geometry}  	
\usepackage{amsfonts}
\usepackage{amsmath} 
\usepackage{centernot}
\usepackage{booktabs}
\usepackage{multirow}
\usepackage{bbm}
\usepackage[]{xcolor}

\usepackage{relsize}

\usepackage[doublespacing]{setspace} 
\usepackage{float}
\usepackage{caption}
\usepackage{graphicx}
\usepackage{adjustbox}
\usepackage{lscape}

\usepackage{adjustbox}

\usepackage{bbm}
\usepackage{subfig}

\usepackage{xr-hyper}

\newcommand\independent{\protect\mathpalette{\protect\independenT}{\perp}}
\def\independenT#1#2{\mathrel{\rlap{$#1#2$}\mkern2mu{#1#2}}}

\usepackage{tikz}
\usetikzlibrary{positioning}
\usetikzlibrary{calc}
\usetikzlibrary{arrows}
\usetikzlibrary{decorations.markings}

\usepackage{authblk}

\usepackage{cite}

\DeclareMathOperator{\E}{\mbox{E}}

\definecolor{forestgreen}{RGB}{34,139,34}
\usepackage{hyperref}
\hypersetup{
    colorlinks=true,
    linkcolor=magenta,
    filecolor=magenta, 
    citecolor=forestgreen,      
    urlcolor=blue
}

\usepackage[utf8]{inputenc}
\DeclareUnicodeCharacter{200E}{}

\usepackage{datetime}

\def\paperversionmajor{7}
\def\paperversionminor{0}

\makeatletter
\newcommand*{\addFileDependency}[1]{
  \typeout{(#1)}
  \@addtofilelist{#1}
  \IfFileExists{#1}{}{\typeout{No file #1.}}
}
\makeatother

\newcommand*{\myexternaldocument}[1]{%
    \externaldocument{#1}%
    \addFileDependency{#1.tex}%
    \addFileDependency{#1.aux}%
}

\myexternaldocument{z_appendix}

\begin{document}


\title{%
Center-specific causal inference with multicenter trials \\
\Large Reinterpreting trial evidence in the context of each participating center}



\author[1,2]{Sarah E. Robertson\thanks{Address for correspondence: Sarah E. Robertson; Box G-S121-8; Brown University School of Public Health, Providence, RI 02912; email: \texttt{sarah\_robertson@brown.edu}; phone: (401) 863-1000.}}
\author[3]{Jon A. Steingrimsson}
\author[4]{Nina R. Joyce}
\author[5]{Elizabeth A. Stuart}
\author[6]{Issa J. Dahabreh}

\affil[1]{Center for Evidence Synthesis in Health, Brown University School of Public Health, Providence, RI}
\affil[2]{Department of Health Services, Policy \& Practice, Brown University School of Public Health, Providence, RI}
\affil[3]{Department of Biostatistics, Brown University School of Public Health, Providence, RI }
\affil[4]{Department of Epidemiology, Brown University School of Public Health, Providence, RI}
\affil[5]{Departments of Mental Health, Biostatistics, and Health Policy and Management, Johns Hopkins Bloomberg School of Public Health, Baltimore, MD}
\affil[6]{Departments of Epidemiology and Biostatistics, Harvard School of Public Health, Boston, MA}


\maketitle{}

\thispagestyle{empty}

\noindent

\vspace{0.1in}
\noindent

\vspace{0.1in}
\noindent

\vspace{0.1in}
\noindent

\clearpage
\noindent

\vspace{0.1in}
\noindent

\thispagestyle{empty}

\clearpage
\setcounter{page}{1}

\hypersetup{pageanchor=true}

\vspace*{0.25in}

\abstract{
\noindent
In multicenter randomized trials, when effect modifiers have a different distribution across centers, comparisons between treatment groups that average over centers may not apply to any of the populations underlying the individual centers. Here, we describe methods for reinterpreting the evidence produced by a multicenter trial in the context of the population underlying each center. We describe how to identify center-specific effects under identifiability conditions that are largely supported by the study design and when associations between center membership and the outcome may be present, given baseline covariates and treatment (``center-outcome associations''). We then consider an additional condition of no center-outcome associations given baseline covariates and treatment. We show that this condition can be assessed using the trial data; when it holds, center-specific treatment effects can be estimated using analyses that completely pool information across centers. We propose methods for estimating center-specific average treatment effects, when center-outcome associations may be present and when they are absent, and describe approaches for assessing whether center-specific treatment effects are homogeneous. We evaluate the performance of the methods in a simulation study and illustrate their implementation using data from the Hepatitis C Antiviral Long-Term Treatment Against Cirrhosis trial.}

\clearpage
\clearpage
\section{Background}
\setcounter{page}{1}

Multicenter trials (also know as multisite individually randomized trials \cite{miratrix2020applied}) can generate robust, clinically relevant evidence by attaining sample sizes that are large enough to precisely estimate modest effects and by enrolling individuals from diverse populations \cite{baigent1997need, fleiss1986analysis, raudenbush2000statistical}. Most major clinical trials in the past decade have had a multicenter design \cite{califf2012characteristics}.

Analyses of multicenter trials typically produce a single overall effect estimate, under the assumption that the treatment effect does not vary across centers or that a sample-size weighted average of center-specific effects is meaningful. Yet, the sample of participants in a multicenter trial is unlikely to reflect any actual underlying population and instead is affected by the choice of centers and their recruitment practices. Consequently, when the assumption of a common treatment effect across centers also does not hold, that is, when the treatment effect varies over effect modifiers (moderators) that have a different distribution across centers, the overall effect estimate across centers may not apply to any of the populations underlying the centers \cite{fleiss1986analysis}. Differences in the distribution of effect modifiers across centers might occur because centers serve different underlying populations (e.g., because referral patterns result in case-mix variability) \cite{ berkowitz2019detecting}. Therefore, investigators at centers participating in the trial will often want to reinterpret the evidence produced by the multicenter trial in the context of the population underlying their own center \cite{orr2019using}. 

Methods for analyzing multicenter trials that do not assume a common average treatment effect across centers, and instead estimate center-specific treatment effects, are available \cite{senn2008statistical}. The most common approach is to estimate center-specific effects using an unadjusted comparison of the outcome between treatment groups, separately for each center. This approach is straightforward, but when estimating the treatment effect in a particular center it ignores the evidence from all other centers. To improve precision, it may be desirable to use information across centers. For example, it is possible to combine information from different centers to estimate center-specific effects with various ``random effects'' methods \cite{gould1998multi, louis1991using}. These methods rely on the assumption that centers (or, at least, center effects) are sampled from a population \cite{senn2008statistical, higgins2009re}. Modeling the distribution of effects requires unidentifiable assumptions \cite{hubbard2010gee} and is particularly challenging when the number of centers is small. When incorporating covariate information, these methods estimate conditional average treatment effects and require correct model specification for valid inference.

Here, we describe methods for estimating center-specific effects when the centers are assumed to be fixed. Our approach allows for heterogeneous center-specific average treatment effect and does not assume that centers (or effects) are sampled from a population. Unlike conventional outcome regressions with fixed center effects, we allow for effect modification by baseline covariates and marginalize over the distribution of those covariates to obtain center-specific average treatment effects. We provide two main results, one that is valid in the presence of associations between center membership and the outcome given covariates and treatment (we refer to these associations as ``center-outcome associations''), and one that is valid under a condition of ``no center-outcome associations,'' when the outcome is independent of the center membership given covariates and treatment. When center-outcome associations may be present, we show how center-specific treatment effects can be identified, conditional on baseline covariates and treatment, under identifiability conditions that are largely supported by study design. When there are no center-outcome associations given baseline covariates and treatment, center-specific treatment effects can be estimated using analyses that combine information across centers. We propose methods for estimating center-specific average treatment effects and describe approaches for assessing whether these effects are homogeneous. Our proposed estimators can exploit available treatment, covariate, and outcome information across centers, are robust to model misspecification, and allow for the use of machine learning methods for estimating models of the outcome and, when needed, the probability of participation in each center. We evaluate the performance of the methods in a simulation study and illustrate the implementation of the methods using data from the Hepatitis C Antiviral Long-Term Treatment Against Cirrhosis trial.

\section{Study design and causal quantities of interest}

We use individual-level data from randomized individuals in a collection of centers $\mathcal{C}=\{1, \cdots, m \}$ from a multicenter trial; here, $m$ is the total number of centers. The trial involves random assignment of treatment $A$, for each individual, that takes values in a finite set $\mathcal{A}$ to estimate the effects of treatment on a (continuous, binary, or count) outcome $Y$ that is assessed at the end of the study.

Our sampling model assumes that the participants in each center can be viewed as a simple random sample from some underlying infinite super-population of individuals that is stratified by $C$ \cite{dahabreh2019study, dahabreh2020toward}. Importantly, the number of centers in $\mathcal C$ is assumed to be fixed (and not growing with sample size), and we allow the sampling fractions from the populations underlying each center to be variable and unknown to the investigators. We describe the sampling model in more detail in Appendix 1. For all randomized individuals we have information on their center membership, $C$, treatment assignment $A$, baseline covariates $X$, and outcome $Y$. Thus, each center $c \in \mathcal C$ contributes information on $(X_i, C_i=c, A_i, Y_i)$, $i=1, \cdots, n_c$,  where $n_c$ denotes the total number of randomized individuals in that center. The total number of randomized individuals across all the centers in the trial is $n=\sum_{c=1}^{m} n_c$. Table \ref{table_data_structure} shows the observed data structure for a multicenter trial. Throughout, we use $I(\cdot)$ to denote the indicator function; italicized capital letters denote random variables and lowercase letters denote realizations of random variables.

We would like to use the trial data to learn about causal quantities in the population underlying each of the trial centers. Let $Y^a$ be the potential (counterfactual) outcome under intervention to set treatment $A$ to $a \in \mathcal{A}$ \cite{Rubin1974, holland1986statistics, robins2000d}. For simplicity we assume complete adherence to the assigned treatment in the trial and no loss to followup, but the methods we discuss can be extended to account for non-adherence and incomplete followup \cite{dahabreh2019generalizing, dahabreh2020toward}. We are interested in the center-specific average treatment effect, $\E[Y^a-Y^{a'}| C = c] =  \E[Y^{a} | C=c ]  - \E[Y^{a'}|C=c]$, for each $c \in \mathcal C$ and every pair of treatments $a$ and $a^\prime$ in $\mathcal A$. Because the components of the center-specific average treatment effects, that is, the center-specific potential outcome means, $\E[Y^a|C=c]$ for $c \in \mathcal C$, are of inherent scientific interest, in the remainder of the main text, we focus on their identification and estimation. In Appendix 2 we present identification results for center-specific potential outcome means and center-specific average treatment effects, and we also show that the center-specific treatment effects can be identified and estimated under weaker conditions (which do not, however, suffice to identify the potential outcome means).

\section{Identification}

\subsection{Center-outcome associations may be present} 

\paragraph{Identifiability conditions:} We first consider three fairly ``standard'' conditions that allow us to identify center-specific average treatment effects, without any assumption about the absence of center-outcome associations, conditional on baseline covariates and treatment. 

\vspace{0.05in}
\noindent
\emph{A1 – Consistency of potential outcomes:} if $A_i = a$, then $Y_i^a=Y_i$, for every individual $i$ in the population underlying the randomized trial and each treatment $a \in \mathcal A$.

\vspace{0.05in}
\noindent
\emph{A2 – Conditional exchangeability over $A$ in the trial}: for every $a \in \mathcal A$ and every $c \in \mathcal C$, $Y^a \independent A |(X,C = c)$. 

\vspace{0.05in}
\noindent
\emph{A3 – Positivity of treatment assignment in each center}:  
for every covariate pattern $x$ and every center $c$ with positive density $f(x,c) \neq 0$, $\Pr[A=a|X=x,C=c]>0$, for each treatment $a \in \mathcal  A$. 

Conditions A1 through A3 are fairly standard conditions when analyzing a randomized trial \cite{hernan2020causal} and are largely supported by the design of the trial. Condition A1 requires that the intervention is well-defined, in the sense that there is no interference and no outcome relevant treatment variation \cite{vanderWeele2009, hernan2011compound, vanderWeele2013causal}, and that center membership does not affect the outcome, except possibly through treatment (an exclusion restriction assumption \cite{dahabreh2019generalizing, dahabreh2020benchmarking}). In Appendix 2 we show how some of our results can be given a causal interpretation even if center membership has effects on the outcome that are not mediated by treatment (e.g., if the centers apply different treatments that directly affect the outcome). 

Condition A2 is a condition of no unmeasured confounding within strata defined by center and baseline covariates. This condition is expected to hold in marginally (unconditionally) randomized trials, where the treatment assignment mechanism does not depend on covariates and and does not vary across centers. It is also expected to hold in conditionally randomized trials where randomization is conditional on covariates within each center, and where the treatment assignment mechanism may vary across centers. In fact, in most multicenter trials treatment assignment is independent of covariates and center membership, such that, by design, $(Y^a,X,C) \independent A$. This condition is stronger than condition A2 because $(Y^a,X,C) \independent A$ implies $Y^a \independent A |(X,C)$; but $Y^a \independent A |(X,C)$ does not necessarily imply $(Y^a,X,C) \independent A$.

Condition A3 simply requires that, within covariate levels, every center randomizes all individuals to each level of treatment with a nonzero probability. This condition is also expected to hold both in marginally and conditionally randomized trials.

\paragraph{Identification:} Under conditions A1 through A3, for $c \in \mathcal C$ and $a \in \mathcal A$, the center-specific potential outcome mean, $\E[Y^a |C=c]$, is identified by the observed data functional 
\begin{equation}\label{id:phi}
    \phi(c, a)= \E \big[ \E[Y|X, C=c, A=a] \big| C=c \big] = \E \big [\E[Y|X, C, A=a] \big | C=c \big ].
\end{equation} 
The inner expectation in this expression conditions on center $C$, among individuals assigned to $A=a$. The outer expectation marginalizes over the covariate distribution in the population underlying the target center of interest, $C=c$. The center-specific average treatment effect comparing treatments $a$ and $a^\prime$ can be identified by $\delta_{\phi}(c, a, a^\prime) =  \phi(c,a)-  \phi(c, a^\prime)$.

As noted above, in most multicenter trials treatment assignment is independent of potential outcomes, covariates, and center membership, $(Y^a,X,C) \independent A$. Under this condition, the causal quantity of interest, $\E[Y^{a} | C=c ]$, is identified by the center- and treatment-specific expectation of the observed outcome (unconditional on covariates), that is,
\begin{equation}\label{id:tau}
    \tau(c, a)= \E[Y|C=c, A=a].
\end{equation} 
The center-specific average treatment effect comparing treatments $a$ and $a^\prime$ can be identified by $\delta_{\tau}(c, a, a^\prime) =  \tau(c,a)-  \tau(c, a^\prime)$.

In the remainder of this paper, we mostly focus on results that hold under the weaker condition A2, such as the result in equation \eqref{id:phi}, rather than require the stronger condition $(Y^a,X,C) \independent A$. Working under the weaker condition means that our results apply more generally, not only to marginally randomized multicenter trials, but also to conditionally randomized multicenter trials, as well as multicenter/pooled analyses of observational studies \cite{blettner1999traditional} with no unmeasured confounding.


\subsection{Center-outcome associations are absent}

\paragraph*{The ``no center-outcome association'' condition:} Consider the following additional independence condition, which is not guaranteed to hold by study design:

\vspace{0.05in}
\noindent
\emph{A4 – No center-outcome associations given covariates and treatment:} for every treatment $a \in \mathcal A$, $Y \independent  C|X, A=a$.

Informally, condition A4 states that the distribution of outcomes does not depend on center participation, given baseline covariates and treatment. Condition A4 can fail, for instance, if there are unmeasured covariates that influence center participation and are also predictors of the outcome (after conditioning on the measured covariates $X$ and treatment). In Appendix 3, we show that condition A4 can be derived from more basic causal identifiability conditions, similar to those invoked in the emerging literature on ``generalizability'' and ``transportability'' analyses \cite{dahabreh2019generalizing, dahabreh2020extending}. The connection with this literature is not crucial for understanding the practical implications of condition A4 for the analysis of multicenter trials.

Note that condition A4 does not imply that the center-specific potential outcome means $\E[Y^a | C = c]$, or average treatment effects $\E[Y^a - Y^{a^\prime} | C = c]$, are equal across centers (i.e., homogeneous over $C$). Even if condition A4 holds, that is to say, even if the outcome is independent of center membership conditional on measured baseline covariates and treatment, the covariates, including effect modifiers and outcome predictors, can have a different distribution across centers. Thus, center-specific potential outcome means and average treatment effects can vary across centers even if condition A4 holds.

Last, note that condition A4 only involves observed variables and thus we can assess whether it is consistent with the data (see references \cite{hartman2013, lu2019causal, dahabreh2020toward, dahabreh2020benchmarking} for similar observations in different contexts). Testing for conditional independence between random variables is a hard problem, in the sense of \cite{shah2020hardness}, and we do not recommend any specific testing procedure for practical application. The performance of any procedure for testing conditional independence will depend on the law underlying the data and the amount of available data. We do note, however, that condition A4 does imply that the expectation of the outcome given covariates and treatment, does not vary across centers. This implication should be assessed in light of substantive knowledge, which can be supplemented by evaluating the following null hypothesis, for each $a \in A$:
\begin{equation*}
    \begin{split}
    H_0: \quad & \E[Y|X,C=1,A = a]  = \cdots =\E[Y|X,C=m,A=a] = \E[Y|X,A=a].
    \end{split}
\end{equation*} 
Many procedures are available for testing equality between conditional expectation functions \cite{delgado1993testing, neumeyer2003nonparametric, racine2006testing, luedtke2019omnibus}. For instance, we can use an analysis of covariance (ANCOVA) test to compare a model for the expectation of the outcome that includes the main effects of treatment and baseline covariates (and possibly their interactions) against a model for the expectation of the outcome that additionally includes center indicators (and possibly any interactions between center indicators, treatment, and baseline covariates) \cite{fleiss1986analysis}. We reiterate that the performance of any such procedure depends on the underlying data law and the amount of available data.

\paragraph*{Identification:} Under conditions A1 through A4, for $c \in \mathcal C$ and $a \in \mathcal A$, the center-specific potential outcome mean $\E[Y^a |C=c]$ is identified by the observed data functional \begin{equation}\label{id:psi}
    \psi(c, a)=\E \big[ \E[Y|X,A=a] \big| C=c \big].
\end{equation}
The inner expectation of this expression does not involve conditioning on $C$ and completely pools information across centers, among individuals assigned to treatment group $A = a$ (conditional on covariates $X$). The outer expectation marginalizes over the distribution of covariates in the target center of interest, $C=c$. Importantly, $\psi(c, a)$ identifies the center-specific potential outcome means only when all conditions, A1 through A4, hold. In that case, complete pooling of the data across centers is reasonable. In contrast, when condition A4 is not plausible or is rejected by the data, identification is still possible under conditions A1 through A3, using analyses that are appropriate in the presence of center-outcome associations. Nevertheless, as we argue in the next section, even in analyses when center-outcome associations may be present and do not completely pool the data across centers, there are opportunities to ``borrow strength'' across centers when estimating the association between covariates or treatment, and the outcome. The center-specific average treatment effect comparing treatments $a$ and $a^\prime$ can be identified by $\delta_{\psi}(c, a, a^\prime) =  \psi(c,a)-  \psi(c, a^\prime)$.

\section{Estimation and inference}

\subsection{Center-outcome associations may be present}\label{sec:present_center_association} 

In Appendix 4, we argue that the following is a reasonable estimator for $\phi(c, a)$: 
\begin{equation}\label{eq:estimatorphi2} 
    \begin{split} 
        \widehat \phi(c, a) = \dfrac{1}{ n_c} \sum\limits_{i=1}^{n} \Bigg\{ \dfrac{I(C_i=c,A_i = a)}{ \widehat e_a(X_i,C_i)} \Big\{ Y_i - \widehat g_a(X_i,C_i) \Big\} + I( C_i = c) \widehat g_a(X_i, C_i) \Bigg\}, 
    \end{split}
\end{equation} 
where $\widehat g_a(X,C)$ is an estimator for $\E[Y |X, C, A = a]$; and $\widehat e_a(X,C)$ is an estimator for $\Pr[A = a| X, C]$. Estimating the probability of treatment is not necessary, because the randomization probability is known in the trial (and each center), but estimating it with a parametric model often improves efficiency by adjusting for chance imbalances in the trial \cite{lunceford2004}. We can estimate center-specific average treatment effects comparing treatments $a$ and $a^\prime$ as $\widehat \delta_{\phi}(c,a,a^\prime) = \widehat \phi(c,a)- \widehat \phi(c, a^\prime)$.

Analysis that allow for center-outcome associations do not preclude using information across centers when estimating the probability of treatment or the conditional expectation of the outcome. For example, if we choose to use a regression of the outcome on the main effect of baseline covariates and center indicators in each treatment group, then we are using information across centers by assuming that covariate associations with the outcome are homogeneous across centers, within each treatment group. In contrast, if we choose to use a separate regression of the outcome on covariates within each treatment group and each center, then we are not allowing for any pooling of information across centers. Assuming that at least some covariate effects are homogeneous across centers may be necessary in trials when some centers have a small number of enrolled individuals.

Adjusted analyses where center-outcome associations may be present (as in \eqref{eq:estimatorphi2}) are usually more efficient than crude analyses that do not use covariate information,
\begin{equation}\label{eq:tau} 
    \begin{split} \widehat \tau(c, a) &= \dfrac{1}{ n_{c,a} } \sum\limits_{i=1}^{n} I( C_i=c, A_i=a) Y_i, 
    \end{split}
\end{equation} 
where $n_{c,a}=\sum_{i=1}^{n} I(C_i=c, A_i=a)$. 
It should also be possible to obtain an estimator based on  $\widehat \phi(c,a)$ that is guaranteed to be more efficient than the unadjusted estimator \cite{colantuoni2015leveraging}. Such crude analyses are reasonable only when treatment assignment does not depend on covariates, at least conditional on center. Clearly, when crude analyses are reasonable, we can estimate center-specific average treatment effects comparing treatments  $a$ and $a^\prime$ as $\widehat \delta_{\tau}(c,a,a^\prime) = \widehat \tau(c, a)- \widehat \tau(c, a^\prime)$.

\subsection{Center-outcome associations are absent}\label{sec:absent_center_association}

Suppose now that background knowledge suggests that the outcome is independent of center, conditional on covariates and treatment (i.e., assumption A4 is plausible) and the data does not provide evidence against it, such that the identification result in equation \eqref{id:psi} is valid, and $\psi(c, a)$ identifies the center-specific treatment effects. In Appendix 4, we argue that the following is a reasonable estimator of $\psi(c, a)$:
\begin{equation}\label{eq:estimatorpsi}
  \widehat \psi(c, a) = \dfrac{1}{ n_c } \sum\limits_{i=1}^{n} \Bigg\{ I(A_i = a) \dfrac{\widehat p_{c}(X_i)}{ \widetilde e_a(X_i)}  \Big\{ Y_i - \widetilde g_a(X_i) \Big\} + I(C_i = c) \widetilde g_a(X_i) \Bigg\},
\end{equation}
where $\widetilde g_a(X)$ is an estimator for $\E[Y | X,A = a]$; $\widehat p_{c}(X)$ is an estimator for $\Pr[C = c| X]$; and $\widetilde e_a(X)$ is an estimator for $\Pr[A = a| X]$. In the vast majority of multicenter trials, the treatment assignment mechanism is common across centers, that is $\Pr[A = a| X, C] = \Pr[A = a| X]$, and known to the investigators, but estimating it with a simple parametric model often improves efficiency and adjusts for chance imbalances in the trial \cite{lunceford2004}. When the treatment assignment mechanism varies across centers estimating $\Pr[A = a| X]$ is more subtle; we discuss issues related to modeling strategies in Section \ref{sec:modeling}. We can estimate center-specific average treatment effects comparing treatments  $a$ and $a^\prime$ as $\widehat \delta_{\psi}(c,a,a^\prime) = \widehat \psi(c, a)- \widehat \psi(c, a^\prime)$.

Note that the potential outcome mean estimators $\widehat \phi(c,a)$ and $\widehat \psi(c,a)$ (and the corresponding treatment effect estimators) differ in two important ways: First, they rely on estimation of different nuisance functions. In particular, $\widehat \phi(c,a)$ uses estimation of nuisance functions that are stratified by center, while $\widehat \psi(c,a)$ does not. Furthermore, $\widehat \psi(c,a)$ requires estimation of the probability of center membership, while $\widehat \phi(c,a)$ does not. Second, the two estimators average over different parts of the available data: $\widehat \phi(c,a)$ only sums contributions from the center of interest because all terms in the summand are multiplied by the indicator $I(C = c)$; in contrast, $\widehat \psi(c,a)$ sums the terms $I(A_i = a) \dfrac{\widehat p_{c}(X_i)}{ \widehat e_a(X_i)} \big\{ Y_i - \widehat g_a(X_i) \big\}$ over all centers.

In Appendix 5, we describe additional g-formula and weighting estimators; these estimators can be viewed as special cases of the estimators presented above \cite{dahabreh2019relation} and are not discussed further.

\subsection{Modeling}\label{sec:modeling}

The estimators in the previous section require the estimation of nuisance functions for the expectation of the outcome, in both $\widehat \phi(c,a)$ and $\widehat \psi(c,a)$; the probability of center membership, in $\widehat \psi(c,a)$; and possibly the probability of treatment, in both $\widehat \phi(c,a)$ and $\widehat \psi(c,a)$. We now discuss the estimation of these nuisance functions. 

As we show in Appendix 6, $\widehat \phi(c, a)$ and $\widehat \psi(c, a)$ are robust to certain kinds of misspecification when using models for the nuisance functions. First, consider $\widehat \phi(c, a)$. When at least one of the models for the probability of treatment, $\Pr[A=a|X,C]$ or the outcome, $\E[Y|X,C, A = a]$ is correctly specified, $ \widehat \phi(c, a)$ is consistent. Because the model for the probability of treatment can always be correctly specified, it follows that $\widehat \phi(c, a)$ is robust to misspecification of the outcome model. Informally, use of the outcome models when estimating $\phi(c, a)$ is primarily geared toward improving efficiency. Next, consider $\widehat \psi(c, a)$. This estimator is doubly robust \cite{bang2005} in the sense that it is consistent when either the models for center membership $\Pr[C=c|X]$ and the probability of treatment $\Pr[A = a | X]$ are correctly specified, or the model for the expectation of the outcome $\E[Y|X,A=a]$ is correctly specified (but not necessarily all three models). In fact, as we discuss below, in multicenter trials, correct specification of a model for $\Pr[A = a | X]$ is always possible when the model for $\Pr[C=c|X]$ is correctly specified. 

The conditional expectation of the outcome in each treatment group, given baseline covariates and center membership, $\E[Y|X,C, A = a]$, or just given covariates, $\E[Y | X, A = a]$, is never known and has to modeled. Furthermore, modeling $\Pr[A = a | X, C]$ to estimate $\widehat e_a(X, C)$ for use in $\widehat \phi(c,a)$ is straightforward even when the assignment mechanism varies across centers, possibly conditional on covariates. 

Issues related to estimating the probability of treatment across centers, that is $\Pr[A = a |X]$, as required for $\widehat \psi(c,a)$, are a bit more subtle, even though in multicenter trials, the probability of treatment in each center is under the investigators' control. First, consider the most common case where the treatment assignment mechanism is the same across centers and does not depend on covariates. In this common case, $\Pr[A = a | X] = \Pr[A = a]$, and it should be easy to specify a model for the relationship between treatment $A$ and covariates $X$ across centers because any reasonable specification is correct. Next, consider the case where the treatment assignment mechanism is not the same across centers, and may depend on covariates. In this case, we need to model the relationship between treatment and covariates using the pooled data. Alternatively, using the fact that $\Pr[A=a|X] = \sum_{c \in \mathcal C} \Pr[A=a|X, C = c] \Pr[C = c | X]$, we can exploit the known by design center-specific probabilities of treatment and rely instead on modeling the probability of center participation. Specifically, we can write $$\dfrac{\widehat p_{c}(X)}{ \widetilde e_a(X)} = \dfrac{\widehat p_{c}(X)}{\sum_{c^\prime \in \mathcal C}  \widetilde e_{a,c^\prime}(X) \widehat p_{c^\prime}(X)},$$ where $\widetilde e_{a,c^\prime}(X)$ is an estimator for $\Pr[A = a | X, C =c^\prime]$ for each $c^\prime \in \mathcal C$. In multicenter trials the center-specific probability of treatment is known to the investigators; and in most cases the probability does not depend on covariates (only center), that is, $\Pr[A = a | X, C =c^\prime] = \Pr[A = a | C =c^\prime]$. Thus, estimation of the $\widetilde e_{a,c^\prime}(X)$ terms requires no modeling assumptions and the only challenge is modeling the probability of center participation, which is used in the $\widehat p_{c}(X)$ and $\widehat p_{c^\prime}(X)$ terms. The displayed equation above highlights that in a multicenter trial, when using $\widehat \psi(c,a)$, double robustness depends on correct specification of a model for the probability of center membership $\Pr[C=c|X]$ or the expectation of the outcome $\E[Y | X, A = a]$, but not on correct specification of a model for $\Pr[A = a | X]$, because the latter can be written in a way that involves the known center-specific probabilities of treatment $\Pr[A = a | X, C = c]$ and the probability of center membership. For pooled analyses of observational studies, we expect the treatment assignment mechanism to vary and be unknown across studies making it necessary to use modeling in order to estimate $\Pr[A = a | X, C = c]$ for each study $c \in \mathcal C$; consequently, estimating $\Pr[A=a|X]$ in pooled analyses of observational studies can prove particularly challenging.

The most common approach for estimating the expectation of the outcome, the probability of center membership, and (if needed) the probability of treatment, is to posit parametric (finite dimensional) models, often modeling center membership using fixed or random effects \cite{fleiss1986analysis, gelman2006data, senn1998some}. Yet, substantive knowledge typically is not sharp enough to ensure that the models for the expectation of the outcome and the probability of participation are correctly specified. This issue is most pressing for the estimator $\widehat \psi(c, a)$ which requires correct specification of at least one of the models for the probability of center membership $\Pr[C=c|X]$ or the expectation of the outcome $\E[Y|X,A=a]$, and does not enjoy the full robustness of $\widehat \phi(c, a)$. Thus, when estimating the nuisance functions, we may want to use flexible (data-adaptive) machine learning methods to reduce the risk of model misspecification. Machine learning estimators of the nuisance functions typically have rates of convergence that are slower than $\sqrt{n}$, resulting in estimators of causal quantities that are not $\sqrt{n}$-consistent (e.g., this is true for the g-formula and weighting estimators that we give in Appendix 5). Still, because $\widehat \phi(c, a)$ and $\widehat \psi(c, a)$ rely on the efficient influence function for their corresponding functional (shown in Appendix 6), they will remain $\sqrt{n}$-consistent when using a wide array of machine learning methods that may converge slower than a $\sqrt{n}$-rate \cite{chernozhukov2017double}. In addition, the estimators can be easily modified to use sample splitting and cross-fitting approaches for double/debiased machine learning \cite{chernozhukov2017double}, in order to further  weaken the conditions required for the estimators to be $\sqrt{n}$-consistent.

\subsection{Inference}
To construct confidence intervals for the estimators $\widehat \phi(c, a)$ and $\widehat \psi(c, a)$ (and the corresponding average treatment effect estimators), we can use the influence curve based estimator of the standard error \cite{vanderLaan2003}, which we describe further in Appendix 7. When using parametric models, a better approach may be to use the sandwich estimator \cite{stefanski2002} or resampling-based methods (e.g., jackknife or bootstrap) \cite{efron1994introduction} to obtain standard errors for estimators of potential outcome means or treatment effects.

\subsection{Efficiency considerations}
In Appendix 6, we show that the asymptotic variance bound for regular estimators of $\psi(c,a)$ is less than or equal to the asymptotic variance bound for regular estimators of $\phi(c,a)$. To understand the practical implications of this result consider two groups of investigators that differ in their willingness to rely on assumption A4: the first group finds Assumption A4 plausible (and consistent with the data) and is thus willing to rely on assumptions A1 through A4; this group can use the estimator $\widehat \psi(c,a)$. The second group is only willing to rely on assumptions A1 through A3 and can use the estimator $\widehat \phi(c,a)$. If the investigators in the first group are correct and condition A4 holds, then, in large samples, they will generally be able to draw more precise inferences than the investigators in the second group -- provided both groups are able to specify approximately correct models for the nuisance functions required by their respective estimators and estimate those models at sufficiently fast rates, in order for the estimators to attain the asymptotic variance bound. In Appendix 8, we describe the relationship between our estimators and the estimators used in transportability analyses, and show that the latter may be less efficient.

\section{Assessing homogeneity of center-specific treatment effects}

An important consideration in multicenter trials is whether the center-specific treatment effects comparing treatments $a$ and $a^\prime$ are constant across centers, that is, whether the following null hypothesis holds:  
$$H_0: \quad \E[Y^a - Y^{a^\prime} |C=1]=\E[Y^a - Y^{a^\prime} |C=2]= \ldots = \E[Y^{a} - Y^{a^\prime} |C=m].$$
If the trial is marginally or conditionally randomized on $C$, we can rewrite the hypothesis as:
\begin{equation*}
H_0: \quad \delta_{\tau}(1,a,a^\prime) = \delta_{\tau}(2,a,a^\prime) = \ldots  = \delta_{\tau}(m,a,a^\prime).
\end{equation*}
We can assess this hypothesis, for example, using ANOVA to compare a model with product terms between treatment and the center indicators and a model without these product terms. 

In order to use covariate information, we can exploit conditions A1 through A3, to rewrite the null hypothesis as follows:
\begin{equation*}
H_0: \quad \delta_{\phi}(1,a,a^\prime) = \delta_{\phi}(2,a,a^\prime) = \ldots  = \delta_{\phi}(m,a,a^\prime).
\end{equation*}
We can assess this hypothesis using the estimator $\widehat \delta_{\phi}(c,a,a^\prime)$, for each $c \in \mathcal C$, to estimate each of the contrasts involved and assess their equality.

Furthermore, if condition A4 also holds, then the null hypothesis can be written as follows:
\begin{equation*}
H_0: \quad \delta_{\psi}(1,a,a^\prime) = \delta_{\psi}(2,a,a^\prime) = \ldots  = \delta_{\psi}(m,a,a^\prime).
\end{equation*}
We can assess this hypothesis using the estimator $\widehat \delta_{\psi}(c,a,a^\prime)$, for each $c \in \mathcal C$,  to estimate each of the contrasts involved and assess their equality. The last two forms of the null hypothesis can be used to obtain a valid test of homogeneity for conditionally randomized trials and for multicenter/pooled observational studies with no unmeasured confounding. As always, the performance of statistical tests depends on the law underlying the data and the amount of data available.

\section{Simulation study}

We ran a simulation study with 1000 runs to evaluate the finite sample performance of different methods for estimating center-specific average treatment effects.

\paragraph{Data generation:} The simulation study was motivated by the empirical example, the HALT-C trial \cite{di2008prolonged}, that we use in the next section. Specifically, we fitted outcome (platelets at 9 months of follow-up) and center membership models in the trial, conditional on standardized versions of three baseline covariates (platelets, age, white blood cell count) to inform the choice of generative models for the simulation study. Furthermore, our chosen sample size is similar to that of the trial. 

In each simulation, we generated three independent standard normal covariates, $X_j$, $j=1,2,3$ for 1000 observations, similar to the HALT-C trial size. Next, we allocated individuals to one of ten possible centers in the trial using a multinomial logistic model, $$ C | X  \sim \text{Multinomial}\left(p_1, p_2,\cdots, p_{10}; \sum\limits_{c=1}^{10}n_c \right), $$ with
\begin{equation*}
    \begin{split}
p_1 &= \Pr[C= 1 | X] = 1 - \sum_{k=2}^{10} p_{k}, \mbox{ and }  \\
p_k &=\Pr[C = k| X] = \dfrac{\text{exp}(\beta_{k} X^T) }{1+ \sum_{j=1}^{10} \exp(\beta_{j}X^T)}, \mbox{ for } k =2,\ldots,10,
    \end{split}
\end{equation*} 
where $X=(1,X_1, X_2,X_3)$ and $\beta_k= (\beta_{k,0}, \ldots, \beta_{k,3})$ is the vector of regression coefficients for center $k$. We provide the values for the coefficients in Appendix 9. We generated treatment assignment $A$ as a binary random variable with a probability of 0.50. We simulated outcomes from the model $Y = 161 + 62 X_1 - X_2 - X_3 - 43 A - 21 X_1 \times A + e$, where $e$ follows a normal distribution with mean 0 and a standard deviation of 36. To evaluate the effect of stronger selection and effect heterogeneity, we considered a scenario where we doubled the coefficient for the interaction between $X_1$ and $A$ in the outcome model and doubled the coefficients for $X_1$ in the multinomial model for center membership.

\paragraph{Estimators and performance metrics:} We implemented the treatment effect estimators described in Sections \ref{sec:present_center_association} and \ref{sec:absent_center_association}. For comparison, we also considered three different estimators that are commonly used in practice: 1) the coefficient of treatment from an ordinary least squares regression with treatment as the only regressor, completely pooling data across all centers (``pooled''); 2) the coefficient of treatment from an ordinary least squares regression with treatment and fixed center effects as regressors (``FE1''); 3) the coefficient of treatment from an ordinary least squares regression with treatment, covariates ($X_1, X_2, X_3$), and fixed center effects as regressors (``FE2'').

We examined the bias and mean squared error (MSE) of the above estimators. Bias was defined as the difference between the estimated effect using each of the methods and the true center-specific average treatment effect in each center (the latter was obtained numerically by averaging the true conditional expectation over 10 million observations). We also examined coverage for Wald-style confidence intervals centered around $\widehat \delta_{\tau}(c,1,0)$, $\widehat \delta_{\phi}(c,1,0)$, and $\widehat \delta_{\psi}(c,1,0)$. For $\widehat \delta_{\tau}(c,1,0)$, we used the standard error of the coefficient for treatment from the linear regression output; for $\widehat \delta_{\phi}(c,1,0)$ and $\widehat \delta_{\psi}(c,1,0)$ we used influence curve based estimates of the standard error. For comparison, we also calculated the estimated standard error from the sandwich error.

\paragraph{Simulation results:} In the main text, we report the results of the scenario with stronger selection and effect heterogeneity. Bias estimates are reported in Table \ref{Table:sim_bias} and MSE estimates are reported in Table \ref{Table:sim_MSE}. The proposed estimators that explicitly target center-specific average treatment effects -- $\widehat \delta_{\tau}(c,1,0)$, $\widehat \delta_{\phi}(c,1,0)$, and $\widehat \delta_{\psi}(c,1,0)$ -- had nearly no bias. The MSE of $\widehat \delta_{\phi}(c,1,0)$ was approximately half and the MSE of $\widehat \delta_{\psi}(c,1,0)$ was approximately one fourth of the MSE of $\widehat \delta_{\tau}(c,1,0)$. Appendix 9 reports coverage, average confidence interval width, and average standard error results from the simulation. The Wald-style intervals centered around our estimators had (nearly) nominal coverage, despite the relatively small sample size of each center. We found that $\widehat \delta_{\psi}(c,1,0)$ always had the smallest average standard error and average confidence width, followed by $\widehat \delta_{\phi}(c,1,0)$; $\widehat \delta_{\tau}(c,1,0)$ had a much larger standard error and wider confidence intervals. In Appendix 9 we compare the average of the estimated standard errors using the influence curve approach or the sandwich estimator of the variance against the estimated standard deviation of the estimators. The influence curve based estimated standard errors were almost identical to those from the sandwich estimators; the averages of the standard error estimates from these two approaches were approximately equal to the estimated standard deviation of the estimator over 1000 runs of the simulation. 

In comparison, the alternative estimators -- pooled, FE1, and FE2 --  often had substantial bias, except for centers that happened to have center-specific average treatment effects close to the underlying parameter of these procedures (e.g., for $c=4$ in our simulation). These estimators had similar MSE values between them. Often but not always, our proposed estimators had lower MSE compared to these alternative estimators. In fact, for some centers (e.g., for $c=1$ or $c=8$), the pooled, FE1, and FE2 estimators had much higher MSE than $\widehat \delta_{\phi}(c,1,0)$ or $\widehat \delta_{\psi}(c,1,0)$. The pooled, FE1, and FE2 estimators had lower MSE compared to $\widehat \delta_{\phi}(c,1,0)$ or $\widehat \delta_{\psi}(c,1,0)$ for centers that happened to have center-specific average treatment effects close to the underlying parameter of these procedures (e.g., for $c=4$ ). Additional simulation results with a lower magnitude of selection and effect heterogeneity are also reported in Appendix 9.

\section{The HALT-C trial}

\paragraph{Study design and data:} The Hepatitis C Antiviral Long-Term Treatment Against Cirrhosis (HALT-C) trial included 10 centers and enrolled 1050 individuals with chronic hepatitis C and advanced fibrosis who had not responded to previous therapy and randomized them to treatment with peginterferon alfa-2a ($a=1$) versus no treatment ($a=0)$. We used platelet count at 9 months of follow-up as the outcome. We considered the following baseline covariates: baseline platelet count, age, sex, previous use of pegylated interferon, race, white blood cell count, history of injected recreational drugs, ever receiving a transfusion, body mass index, creatine levels, smoking status,  previous use of combination therapy (interferon and ribavirin), diabetes, serum ferritin, hemoglobin, aspartate aminotransferase, ultrasound evidence of splenomegaly, and ever drank. We restricted our analyses to individuals with complete covariate and outcome data ($n=948$). Table 4 summarizes covariate information over all centers, stratified by treatment arm; Tables 1 and 2 in Appendix 10 summarize covariate information stratified by center. 

\paragraph{Estimation, model specification, inference:} We applied the estimators $\widehat \delta_{\tau}(c,1,0)$, $\widehat \delta_{\phi}(c,1,0)$, and $\widehat \delta_{\psi}(c,1,0)$ to the HALT-C data, for $c = 1, \ldots, 10$, to estimate the average treatment effect in the target population underlying each of the participating centers. In this section of the paper, we denote the these groups of estimators as $\widehat \delta_{\tau}$, $\widehat \delta_{\phi}$, and $\widehat \delta_{\psi}$ when no confusion arises by suppressing the indexing by center (always the same 10 centers are considered) and treatments (always binary). For comparison, we also applied the pooled, FE1, and FE2 estimators, described in the simulation study. The analyses we report here are meant to illustrate the methods and should not be clinically interpreted.

To use $\widehat \delta_{\phi}$ and $\widehat \delta_{\psi}$, we modeled the expectation of the outcome with linear regression; the probability of center membership with multinomial logistic regression; and the probability of treatment with binary logistic regression. These models included main effects of all the baseline covariates listed above (models for the nuisance functions needed for $\widehat \delta_{\phi}$ included center fixed effects). We obtained 95\% Wald-style confidence intervals using the sandwich estimator of the variance \cite{stefanski2002, saul2017}.

As a stability analysis, we used generalized additive models (GAMs) to estimate models for the outcome and center membership, using the main effects of all baseline covariates. In these analyses, we always modeled the probability of treatment using a (correctly specified) parametric binary logistic regression. When using GAMs, we obtained bootstrap normal distribution-based confidence intervals using 1000 bootstrap samples.

\paragraph{Assessing the presence of center-outcome associations:} To evaluate whether center-outcome associations are present, we used ANCOVA to compare the expectation of the outcome in a linear regression model that included the main effects of baseline covariates and treatment and all possible baseline covariates and treatment products, against a linear regression model that additionally included the main effect of the center indicators and all product terms between baseline covariates, treatment, and center indicators. 

\paragraph{Assessing homogeneity of treatment effects across centers:} To evaluate whether homogeneity of treatment effects across centers holds, we applied $\widehat \delta_{\phi}$ and $\widehat \delta_{\psi}$ and used omnibus Wald chi-square tests for assessing whether the center-specific treatment effects were homogeneous across centers.

\paragraph{Results:} We did not find evidence to reject the null hypothesis of no center-outcome associations, conditional on baseline covariates and treatment (ANCOVA p-value = 0.42), and we concluded that it was reasonable to use the estimator $\widehat \delta_{\psi}$. Figure \ref{fig:forest_plot_main} shows a forest plot comparing average treatment effects in each center: white squares depict the crude analysis using $\widehat \delta_{\tau}$; grey squares depict the adjusted analysis using $\widehat \delta_{\phi}$; and black squares depict the adjusted analysis using $\widehat \delta_{\psi}$. Estimates from $\widehat \delta_{\psi}$ were the most precise, followed by those from $\widehat \delta_{\phi}$, and those from $\widehat \delta_{\tau}$ were the least precise. Numerical results are presented in Table \ref{Table:ATE_main}; results from analyses using generalized additive models, presented in Appendix 10, were qualitatively similar. The pooled, FE1, and FE2 estimates (95\% confidence interval) were -42.3 (-50.5, -34.0), -42.3 (-50.5, -34.0), and -42.4 (-47.3, -37.5), respectively. We did not find evidence to reject the null hypothesis of homogeneity of treatment effects across centers in the analyses using $\widehat \delta_{\phi}$ or $\widehat \delta_{\psi})$ (Wald test p-value, 0.51 and 0.46, respectively).

\section{Discussion}

New infrastructures for data collection will make large multicenter trials more common in the future. For example, the Clinical Trials Transformation Initiative, a public-private partnership between the FDA and 60 other organizations, aims to encourage the conduct of large, simple randomized trials (``megatrials'') \cite{eapen2014imperative}. Embedding multicenter trials in learning healthcare systems \cite{fiore2016} and using new data collection strategies (e.g., electronic health record data capture) substantially reduces the cost of running trials and improves the efficiency of data collection from diverse populations \cite{olsen2007learning}.

Increasing diversity in the data can support analyses that are more broadly applicable, but requires careful specification of the inferential target population to avoid producing average results that do not apply to any real-world setting. Here, we focused on the case where the inferential target populations are the population underlying the centers in a multicenter randomized trial. In this setting, we described methods for reinterpreting the evidence produced by a multicenter trial to obtain center-specific effects. Our approach relies on fixed effects models for a finite collection of centers, avoiding assumptions about sampling from an infinite population of centers.

The methods we propose are most useful when recruitment from the population underlying each center is not selective, such that the center-specific sample of randomized individuals represents a meaningful target population. For example, in multicenter social experiments \cite{orr1999social, bell2016social} or large-scale randomized clinical trials embedded in healthcare systems \cite{fiore2016}, participants can reasonably be viewed as a random sample of some underlying center-specific population. When recruitment into each center is selective, center-specific samples may not reflect the center-specific underlying populations of trial-eligible individuals, and the methods we propose here may be most useful for identifying differences among the centers themselves. For example, the methods may be applied to detect outlying centers (with possible applications for fraud detection), or to explore whether center-specific practices, such as ancillary non-randomized interventions and local policies \cite{orr2019using}, correlate with the magnitude of treatment effects \cite{senn2008statistical}. Alternatively, if a sample from a well-defined target population can be obtained separately from the multicenter trial sample, investigators can also consider transportability methods for meta-analyses, which can be easily modified for use with multicenter trial data, but require additional untestable assumptions about the exchangeability of trial participants and members of the target population \cite{dahabreh2020toward, dahabreh2019efficient}. 

In summary, we described robust and efficient methods for reinterpreting the evidence produced by a multicenter trial in the context of the population underlying each participating center. The methods allow investigators to use the totality of the evidence in the trial for inference about center-specific treatment effects, and to assess whether these effects are heterogeneous across centers. The methods may be useful additions to standard analytical approaches for multicenter trials, when treatment effect modifiers have a different distribution across centers.

\section*{Acknowledgements}
This work was supported in part by Agency for Healthcare Research and Quality (AHRQ) award R36HS028373-01 (Robertson) and by Patient-Centered Outcomes Research Institute (PCORI) award ME-1502-27794 (Dahabreh).  \\
The data analyses in our paper used HALT-C data obtained from the National Institute of Diabetes and Digestive and Kidney Diseases (NIDDK) Central Repositories. The HALT-C trial data are not publicly available, but they can be obtained from the National Institute of Diabetes and Digestive and Kidney Diseases (NIDDK) Central Repository (\url{https://repository.niddk.nih.gov/studies/halt-c/}).
\\
The content of this paper is solely the responsibility of the authors and does not necessarily represent the views of AHRQ, PCORI, the PCORI Board of Governors,the PCORI Methodology Committee, the HALT-C study, the NIDDK Central Repositories, or the NIDDK.


\newpage
\clearpage
\bibliographystyle{ieeetr}
\bibliography{ref_multicenter}

\newpage
\section{Tables and Figures}

\begin{table}[!ht]
	\renewcommand{\arraystretch}{1.3}
	\centering
\caption{Data structure for treatment, $A$, observed outcome, $Y,$ and baseline covariates, $X$, in a multicenter trial. The observed data include $(X, C, A, Y)$ information from $n$ trial participants, of whom $n_c$ are in center $c \in \mathcal{C}$.}\label{table_data_structure}
\begin{tabular}{|ccccc|}
\hline
Unit, $i$                               & $X_i$               & $C_i$      & $A_i$               & $Y_i$               \\ \hline 
\multicolumn{1}{|c}{$1 $}          & $X_{1}$           & 1        & $A_{1}$           & $Y_{1} $          \\
\multicolumn{1}{|c}{$\vdots$}      & $\vdots$          & $\vdots$        & $\vdots$          & $\vdots$          \\
\multicolumn{1}{|c}{$n_1 $}        & $X_{n_1}$         & 1        & $A_{n_1}$         & $Y_{n_1} $        \\ \hline
\multicolumn{1}{|c}{$n_1 + 1 $}      & $X_{n_1 + 1 }$       & 2        & $A_{n_1 + 1 }$       & $Y_{n_1 + 1} $      \\
\multicolumn{1}{|c}{$\vdots$}      & $\vdots$          & $\vdots$        & $\vdots$          & $\vdots$          \\
\multicolumn{1}{|c}{$n_1 + n_2$}        & $X_{n_1 + n_2}$         & 2        & $A_{n_1 + n_2}$         & $Y_{n_1 + n_2} $        \\ \hline
\multicolumn{1}{c}{$\vdots$}      & $\vdots$          & $\vdots$ & $\vdots$          & \multicolumn{1}{c}{$\vdots$}          \\ \hline 
\multicolumn{1}{|c}{$\sum_{c=1}^{m-1} n_c + 1$} & $X_{\sum_{c=1}^{m-1} n_c + 1}$ & $m$      & $A_{\sum_{c=1}^{m-1} n_c + 1}$ & $Y_{\sum_{c=1}^{m-1} n_c + 1}$ \\
\multicolumn{1}{|c}{$\vdots$}      & $\vdots$          & $\vdots$ & $\vdots$          & $\vdots$          \\
$\sum_{c=1}^{m-1} n_c + n_m = \sum_{c=1}^{m} n_c = n$           & $X_{n}$           & $m$      & $A_{n}$           & $Y_{n}$           \\ \hline
\end{tabular}
\end{table}

\begin{table}[ht]
\centering
 \caption{Bias results from the simulation for each method.}
\label{Table:sim_bias}
\begin{tabular}{@{}cccccccc@{}}
\toprule
Center, $c$ & Average $n_c$ & $\widehat \delta_{\tau}(c,1,0)$ & $\widehat \delta_{\phi}(c,1,0)$ & $\widehat \delta_{\psi}(c,1,0)$  & Pooled & FE1 & FE2  \\ \midrule
1 & 57 & -0.31 & -0.25 & -0.16 & 22.60 & 22.61 & 22.58 \\
2 & 100 & 0.41 & 0.24 & -0.15 & -5.10 & -5.09 & -5.12 \\
3 & 135 & -0.47 & -0.38 & -0.16 & 8.83 & 8.84 & 8.81 \\
4 & 68 & -0.04 & -0.24 & -0.13 & -1.92 & -1.91 & -1.94 \\
5 & 80 & -0.53 & -0.12 & -0.14 & 12.97 & 12.98 & 12.95 \\
6 & 107 & -0.42 & -0.18 & -0.13 & -13.57 & -13.56 & -13.60 \\
7 & 94 & 0.43 & 0.36 & -0.18 & -9.76 & -9.75 & -9.78 \\
8 & 110 & -0.04 & -0.16 & -0.14 & -17.57 & -17.56 & -17.59 \\
9 & 43 & -0.67 & -0.60 & -0.13 & -8.12 & -8.11 & -8.14 \\
10 & 206 & 0.03 & -0.26 & -0.14 & 8.02 & 8.03 & 7.99 \\ \bottomrule
\end{tabular}
\caption*{Average $n_c$ is the average sample size in center $c$ across simulations. See the text for the definitions of the pooled, FE1, and FE2 estimators.}
\end{table}

\begin{table}[ht]
\centering
 \caption{MSE results from the simulation for each method.}
\label{Table:sim_MSE}
\begin{tabular}{@{}cccccccc@{}}
\toprule
Center, $c$ & Average $n_c$ & $\widehat \delta_{\tau}(c,1,0)$ & $\widehat \delta_{\phi}(c,1,0)$ & $\widehat \delta_{\psi}(c,1,0)$  & Pooled & FE1 & FE2  \\ \midrule
1 & 57 & 218.01 & 116.98 & 35.97 & 524.42 & 524.78 & 516.55 \\
2 & 100 & 127.63 & 64.19 & 22.51 & 39.58 & 39.40 & 32.94 \\
3 & 135 & 97.58 & 53.03 & 19.13 & 91.55 & 91.64 & 84.29 \\
4 & 68 & 199.89 & 105.11 & 31.35 & 17.25 & 17.14 & 10.47 \\
5 & 80 & 162.56 & 88.99 & 27.06 & 181.87 & 182.04 & 174.43 \\
6 & 107 & 116.04 & 61.40 & 21.29 & 197.81 & 197.47 & 191.54 \\
7 & 94 & 140.35 & 75.18 & 24.11 & 108.84 & 108.57 & 102.40 \\
8 & 110 & 121.16 & 60.01 & 20.88 & 322.19 & 321.77 & 316.10 \\
9 & 43 & 327.29 & 156.03 & 40.93 & 79.46 & 79.22 & 72.95 \\
10 & 206 & 60.72 & 32.17 & 13.12 & 77.82 & 77.90 & 70.60 \\ \bottomrule
\end{tabular}
\caption*{Average $n_c$ is the average sample size in center $c$ across simulations. See the text for the definitions of the pooled, FE1, and FE2 estimators.}
\end{table}

\clearpage
\begin{table}[!ht]
\centering
 \caption{Baseline characteristics in the HALT-C trial, stratified by treatment assignment.}
    \label{Table:haltc_baseline}
\begin{tabular}{@{}lcc@{}}
\toprule
                                                      & $A=1$           & $A=0$           \\ \midrule
Number of individuals                                 & 468             & 480             \\
Baseline platelets, $\times$ 1000/mm$^3$              & 165.5 (62.5)  & 164.7 (68.1)  \\
Age in years                                          & 51.2 (7.4)    & 50.0 (7.2)    \\
Female                                                & 138 (29.5)      & 133 (27.7\%)    \\
Received pegylated interferon                         & 129 (27.6\%)    & 143 (29.8\%)    \\
White                                                 & 335 (71.6\%)    & 340 (70.8\%)    \\
Baseline white blood cell count, $\times$ 1000/mm$^3$ & 5.9 (1.9)     & 5.6 (1.8)     \\
Used recreational drugs                               & 221 (47.2\%)    & 208 (43.3\%)    \\
Received a transfusion                                & 195 (41.7\%)    & 180 (37.5\%)    \\
Body mass index, weight (kg)/height(m)$^2$            & 29.8 (5.3)    & 30.1 (5.7)    \\
Creatinine, mg/dl                                     & 0.9 (0.2)     & 0.9 (0.2)     \\
Ever smoked                                           & 350 (74.8\%)    & 363 (75.6\%)    \\
Received interferon and ribavirin                     & 389 (83.1\%)    & 389 (81.0\%)    \\
Reported diabetes                                     & 85 (18.2\%)     & 81 (16.9\%)     \\
Serum ferritin, ng/ml                                 & 361.8 (426.9) & 388.6 (433.1) \\
Ultrasound evidence of splenomegaly                   & 156 (33.3\%)    & 163 (34.0\%)    \\
Ever drank alcohol                                    & 388 (82.9\%)    & 396 (82.5\%)    \\
Hemoglobin, g/dl                                      & 15.0 (1.4)    & 15.0 (1.4)    \\
Aspartate aminotransferase, U/L                       & 88.2 (59.7)   & 88.0 (58.8)   \\ \bottomrule
\end{tabular}
\caption*{Results reported as mean (standard deviation) for continuous variables and count (percentage) for binary variables.\\
$A$, indicates randomization to treatment with peginterferon alfa-2a,  ($A=1$) versus no treatment ($A=0$); kg, kilogram; m, meter; mg, milligram; dl, deciliter; ml, milliliter; g, gram; U/L, units per liter.}
\end{table}

\clearpage
\begin{table}[!ht]
\centering
 \caption{Analysis of the HALT-C trial.} \label{Table:ATE_main}
\begin{tabular}{ccp{.2\textwidth}p{.2\textwidth}p{.2\textwidth}}
\toprule
Center, $c$ & $n_c$ & \hfil$\widehat \delta_{\tau}(c,1,0)$ & \hfil$\widehat \delta_{\phi}(c,1,0)$ & \hfil$\widehat \delta_{\psi}(c,1,0)$  \\ \midrule
1 & 48 & \hfil -29.3 (11.5, -70) \par \hfil {[}20.8{]} & -46.9 (-21.8, -71.9) \par \hfil {[}12.8{]} & \hfil -46 (-35.6, -56.4) \par \hfil  {[}5.3{]} \\
2 & 97 & -42.9 (-16.9, -69) \par \hfil {[}13.3{]} & -40.5 (-28.4, -52.5) \par \hfil  {[}6.1{]} & \hfil -42.9 (-35.9, -49.8) \par \hfil  {[}3.5{]} \\
3 & 130 & -44.2 (-20.9, -67.5) \par \hfil {[}11.9{]} & -35.3 (-24, -46.7) \par \hfil  {[}5.8{]} & \hfil -44.4 (-37.3, -51.5) \par \hfil  {[}3.6{]} \\
4 & 66 & -21 (6.3, -48.3) \par \hfil  {[}13.9{]} & -34.2 (-16.2, -52.3) \par \hfil {[}9.2{]} & \hfil -39.4 (-31.5, -47.4) \par \hfil  {[}4.1{]} \\
5 & 76 & -70.9 (-44.6, -97.3) \par \hfil {[}13.4{]} & -58.3 (-43.1, -73.4) \par \hfil  {[}7.7{]} &  \hfil -45.5 (-37, -53.9) \par \hfil  {[}4.3{]} \\
6 & 101 & -35.6 (-10.7, -60.5) \par \hfil {[}12.7{]} & -35.3 (-19.3, -51.2) \par \hfil  {[}8.2{]} & \hfil -39.4 (-32.8, -46.1) \par \hfil  {[}3.4{]} \\
7 & 89 & -42.8 (-12, -73.6) \par \hfil {[}15.7{]} & -48.1 (-30.6, -65.6) \par \hfil  {[}8.9{]} & \hfil -43.9 (-36.7, -51.1) \par \hfil  {[}3.7{]} \\
8 & 100 & -49.7 (-26.6, -72.9) \par \hfil  {[}11.8{]} & -40.7 (-26.4, -55.1) \par \hfil  {[}7.3{]} & \hfil -37.3 (-30.9, -43.7) \par \hfil  {[}3.3{]} \\
9 & 42 & -25.8 (7.1, -58.8) \par \hfil {[}16.8{]} & -40.8 (-21.2, -60.4) \par \hfil  {[}10{]} & \hfil -40.7 (-31.9, -49.5) \par \hfil  {[}4.5{]} \\
10 & 199 & -42.8 (-25.4, -60.1) \par \hfil {[}8.9{]} & -46.1 (-36, -56.3) \par \hfil  {[}5.2{]} &  \hfil -44.6 (-38.2, -51.1) \par \hfil  {[}3.3{]} \\ \bottomrule
\end{tabular}
\caption*{The number of individuals in center $c$ is $n_c$. The 95\% Wald confidence intervals are in parentheses. The standard error calculated from the sandwich estimator are in brackets. }
\end{table}

\newpage
\begin{figure}[!ht]
    \centering
    \caption{Forest plot of HALT-C trial estimates.}
    \label{fig:forest_plot_main}
    \includegraphics[width=11cm]{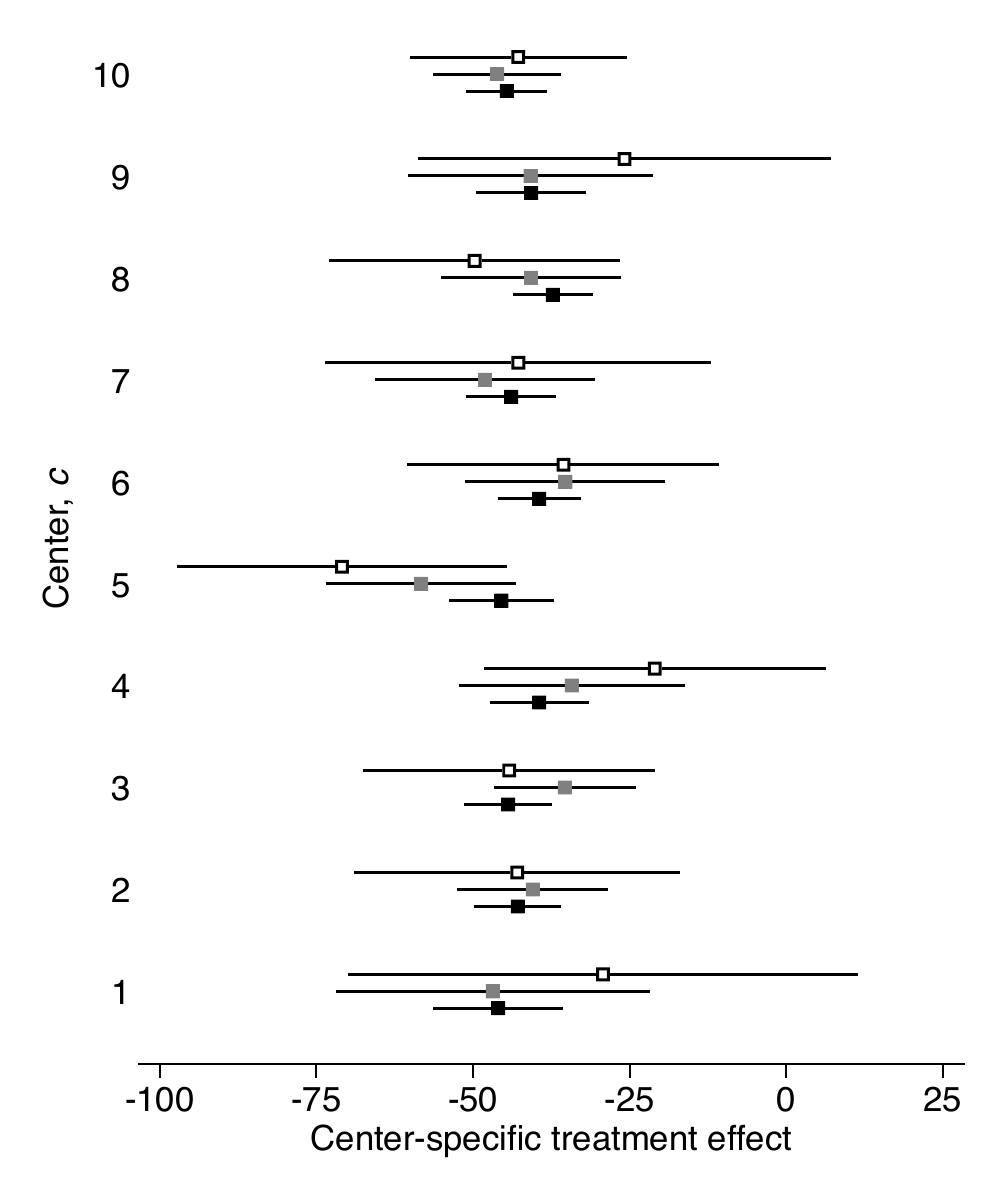}
    \caption*{ATE, center-specific average treatment effects. \\
    Point estimates (square markers) and 95\% confidence intervals (extending lines). White squares represent the crude analysis using the estimator $\widehat \delta_{\tau}(c,1,0)$; grey squares represent the adjusted analysis when center-outcome associations may be present using the estimator $\widehat \delta_{\phi}(c,1,0)$; black squares represent the adjusted analysis when center-outcome associations are absent using the estimator $\widehat \delta_{\psi}(c,1,0)$ .}
\end{figure}

\newpage
\clearpage

\includepdf[pages=-]{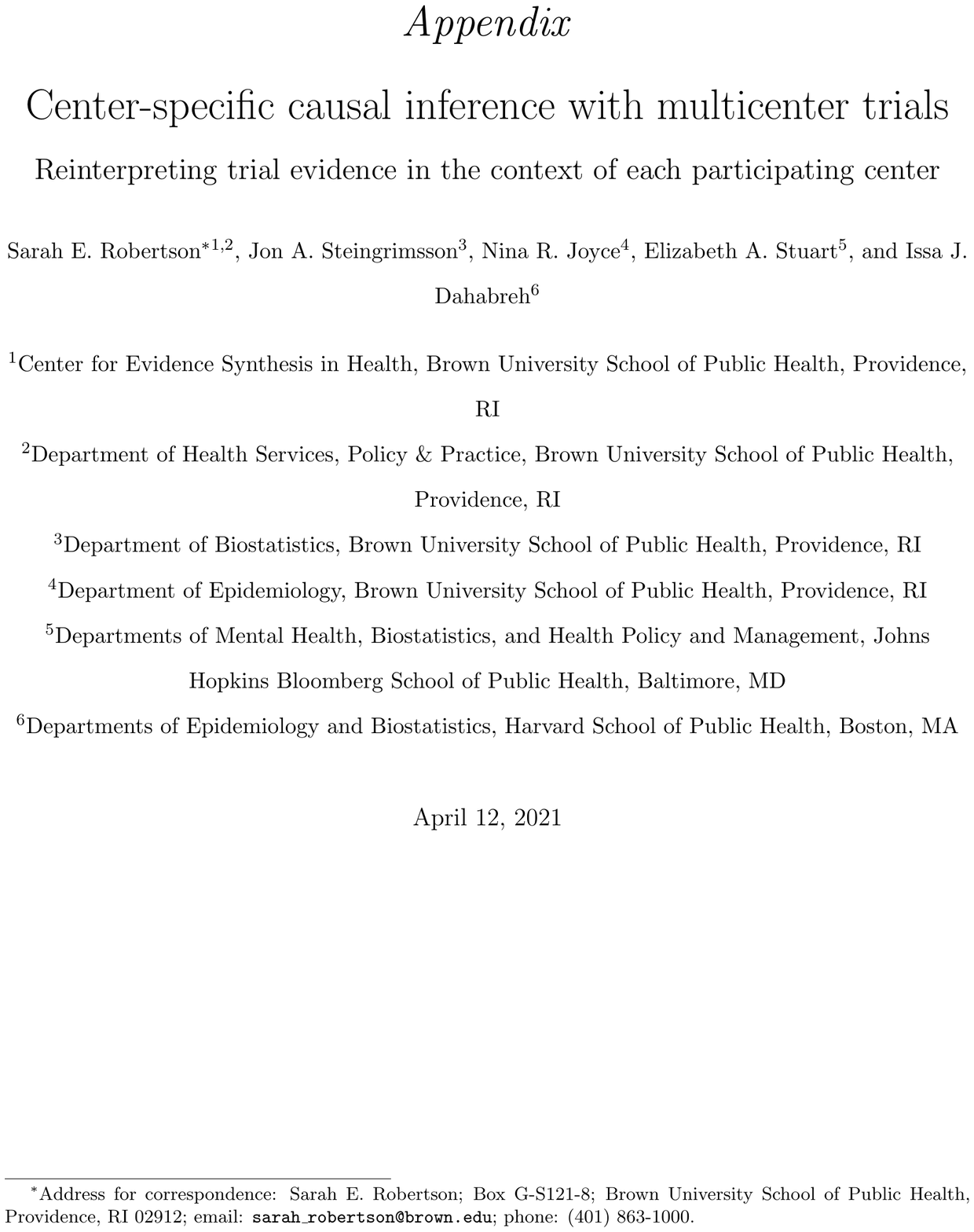}


\ddmmyyyydate 
\newtimeformat{24h60m60s}{\twodigit{\THEHOUR}.\twodigit{\THEMINUTE}.32}
\settimeformat{24h60m60s}
\begin{center}
\vspace{\fill}\ \newline
\textcolor{black}{{\tiny $ $multicenter, $ $ }
{\tiny $ $Date: \today~~ \currenttime $ $ }
{\tiny $ $Revision: \paperversionmajor.\paperversionminor $ $ }}
\end{center}

\end{document}


\maketitle

\thispagestyle{empty}

\clearpage
\setcounter{tocdepth}{1}
\tableofcontents

\thispagestyle{empty}

\clearpage
\appendix 
\renewcommand{\thesection}{Appendix \arabic{section}}
\renewcommand{\thesubsection}{\arabic{section}.\arabic{subsection}}
\pagenumbering{arabic}



\clearpage

\section{Sampling model}\label{appendix:sampling_model}

In order to use information obtained from multiple centers to draw inferences about treatment effects we need a sampling model to reflect the relationships between the data from each center and the underlying population of each center. Similar to our recent work on causally interpretable meta-analysis \citelatex{dahabreh2020toward, dahabreh2019efficient}, we assume that the observed data are obtained by random sampling from an infinite superpopulation of individuals that is stratified by $C$, with sampling fractions that are unknown and possibly unequal constants for each stratum with $C=c$, with $c \in \{0, 1, \cdots, m \}$. We refer to this sampling model as a ``biased sampling model'' \citelatex{bickel1998efficient}, because the sample proportion of individuals in the data with $C=c$, $n^{-1} \sum_{i=1}^{n}I(C_i = c)$, does not in general reflect the superpopulation probability of $C=c$. In the remainder of this appendix all distributions, probability mass functions, and expectations are under the biased sampling model.

\clearpage
\section{Identification}\label{appendix:identification}

\subsection{Identifiability conditions}\label{appendix:identifiability_conditions}

\noindent
\emph{A1 – Consistency of potential outcomes:} if $A_i = a$, then $Y_i^a=Y_i$, for every individual $i$ in the population underlying the randomized trial and each treatment $a \in \mathcal A$

\vspace{0.05in}
\noindent
\emph{A2 – Conditional exchangeability over $A$ in the trial}: for every $a \in \mathcal A$ and every $c \in \mathcal C$, $Y^a \independent A |(X,C = c)$.

\vspace{0.05in}
\noindent
\emph{A3 – Positivity of treatment assignment in each center}:  
for every covariate pattern $x$ and every center $c$ with positive density $f(x,c) \neq 0$, $\Pr[A=a|X=x,C=c]>0$, for each treatment $a \in \mathcal A$. 

\vspace{0.05in}
\noindent
\emph{A4 – No center-outcome associations, given covariates and treatment:} for every treatment $a \in \mathcal A$, $Y \independent  C|X, A=a$.

\subsection{Identification of center-specific potential outcome means}\label{appendix:identification_POM}

We now consider identification results for $\E[Y^a|C=c]$.

\subsubsection{Center-outcome associations may be present} 

\begin{proposition}\label{app:prop1}
Under conditions A1 through A3, \textup{$\E[Y^a |C=c]$} is identified by the observed data functional \textup{$\phi(c, a)=\E \big[\E[Y|X, C=c, A=a] \big|C=c \big]$}, for each treatment $a \in \mathcal A$ and each center $c \in \mathcal C$.
\end{proposition}

\begin{proof}
Starting from the causal quantity of interest,
\begin{equation}\label{id:phi_APP}
	\begin{split}
\E [Y^a|C=c ]	& \stackrel{}{=}  \E \big[\E[ Y^a| X, C=c]\big |C=c] \\
	& \stackrel{}{=}  \E \big[\E[ Y^a|X, C=c, A=a]\big |C=c] \\
	& \stackrel{}{=}   \E \big[ \E[ Y |X, C=c, A=a] |C=c  \big] \\	
	& \stackrel{}{\equiv} \phi(c, a),
	\end{split}
\end{equation}
where the first step follows by the law of iterated expectation, the second by condition A2, the third by condition A1, the last by the definition of $\phi(c, a)$, and quantities conditional on $A = a$ are well-defined by condition A3. 
\end{proof}

In multicenter trials where treatment assignment is independent of covariates and center participation, condition A2 can be replaced by the stronger independence condition $(Y^a,X,C) \independent A$, for each treatment $a \in \mathcal A$. Under this stronger independence condition and conditions A1 and A3, $\E[Y^a |C=c]$ is identified by the observed data functional 
\begin{equation}\label{id:tau_APP}
    \tau(c, a)= \E[Y|C=c, A=a],
\end{equation} 
and it has to be that $ \tau(c, a) = \phi(c, a)$.

\subsubsection{Center –- outcome associations are absent}

\begin{proposition}\label{app:prop2}
Under conditions A1 through A4, \textup{$\E[Y^a |C=c]$} is identified by the observed data functional  \textup{$\psi(c, a)=\E[\E[Y|X,A=a]|C=c]$}.
\end{proposition}

\begin{proof}
As shown in Proposition \ref{app:prop1}, under conditions A1 through A3,
\begin{equation*}
\E [Y^a|C=c ] = \E \big[ \E[ Y |X, C=c, A=a] |C=c  \big].
\end{equation*}
By condition A4, it follows that
\begin{equation}\label{id:psi_APP}
\E [Y^a|C=c ] = \E \big[ \E[ Y |X, A=a] |C=c  \big] \equiv \psi(c, a).
\end{equation}
\end{proof}

\subsection{Identification of center-specific average treatment effects}\label{appendix:identification_ATE}

In addition to the center-specific potential outcome means, we are typically interested in the center-specific average treatment effect, $\E[Y^a-Y^{a'}| C = c] =  \E[Y^{a} | C=c ]  - \E[Y^{a'}|C=c]$, for each $c \in \mathcal C$ and every pair of treatments $a$ and $a^\prime$ in $\mathcal A$. 

\paragraph{Identification under the same assumptions as for center-specific potential outcome means:} When the center-specific potential outcome means are identifiable, we can use contrasts between them to identify center-specific average treatment effects. For example, under assumptions A1 through A3, the center-specific average treatment effect comparing treatments $a$ and $a^\prime$ can be identified by $\delta_{\tau}(c, a, a^\prime) =  \tau(c,a)-  \tau(c, a^\prime)$ or by $\delta_{\phi}(c, a, a^\prime) =  \phi(c,a)-  \phi(c, a^\prime)$. Likewise, under assumptions A1 through A4, the center-specific average treatment effect comparing treatments $a$ and $a^\prime$ can be identified by $\delta_{\psi}(c, a, a^\prime) =  \psi(c,a)-  \psi(c, a^\prime)$.

Next, we discuss the identification of center-specific average treatment effects under weaker assumptions than those needed to identify center-specific potential outcome means. 

\paragraph{Identification under homogeneity of the treatment-outcome association across centers:} 
When center-outcome associations are absent, it is possible to identify the average treatment effect (but not the potential outcome means) by combining conditions A1 through A3, with a condition weaker than condition A4. Consider the following condition:

\vspace{0.05in}
\noindent
\emph{A4$^*$ - Homogeneity of treatment-outcome associations across centers:}
for every pair of treatments $a$ and $a'$ in $\mathcal A$, every center $c \in \mathcal C$, and for every covariate $x$ pattern such that $f(x,c,a) \neq 0$ and $f(x,c,a') \neq 0$, 
\begin{equation*}
    \begin{split}
        &\E[Y|X=x,C=c, A=a] - \E[Y|X=x,C=c, A = a^\prime] \\
            &\quad\quad\quad = \E[Y|X=x, A = a] - \E[Y|X=x, A = a^\prime].
    \end{split}
\end{equation*}

\noindent
Condition A4$^*$ requires homogeneity of the treatment-outcome association across centers (on the risk difference scale), conditional on covariates and treatment. In contrast, condition A4 from the main text, requires ``no center-outcome associations, given covariates and treatment'': for every treatment $a \in \mathcal A$, $Y \independent  C|X, A=a$.
Condition A4$^*$ is weaker than condition A4 in the main text because condition A4 implies condition A4$^*$; but condition A4$^*$ does not imply condition A4. Under condition A4$^*$, the result in Proposition \ref{app:prop2} does not necessarily hold, but we can still obtain an identification result for center-specific average treatment effects.


\begin{proposition}
Under conditions A1 through A3, and A4$^*$, the center-specific average treatment effect, \textup{$\E[Y^a-Y^{a'} |C=c]$} is identified by the observed data functional  \textup{$\rho(c, a, a')\equiv\E[\lambda(a,a';X)|C=c]$, 
where $\lambda(a,a';X) \equiv \E[Y|X, A=a]-\E[Y|X,A=a'].$}
\end{proposition}

\begin{proof}

We will first show that under conditions A1 through A3, and A4$^*$, for each $c \in \mathcal C$, and each $x$ with positive density $f(x,c)>0$,  $\E [Y^a-Y^{a'}|X = x, C=c ] = \lambda(a, a'; x)$. 

Starting from the left-hand-side, 
\begin{equation*}
	\begin{split}
\E [Y^a-Y^{a'}|X = x, C=c ]	
	& \stackrel{}{=}  \E[ Y^a| X = x, C=c,  A=a] - \E[ Y^{a'}| X = x, C=c, A=a']   \\	
		& \stackrel{}{=}  \E[ Y| X = x, C=c, A=a] - \E[ Y| X = x, C=c, A=a'] \\
			& \stackrel{}{=}  \E[ Y| X = x, A=a] - \E[ Y| X = x, A=a'] \\
			& \stackrel{}{\equiv} \lambda(a, a'; x),
	\end{split}
\end{equation*}
where the first step follows from linearity of expectations and condition A2, the second by condition A1, the third by condition A4$^*$, the last by definition of $\lambda(a, a'; X)$ at $X = x$, and quantities conditional on $A=a$ are well-defined by condition A3. 

Using the above result and the law of iterated expectation, we obtain:
\begin{equation*}
	\begin{split}
\E [Y^a-Y^{a'}| C=c ]	
& \stackrel{}{=} \E \big[\E[ Y^a-Y^{a'}| X, C=c]\big |C=c]  \\
	& \stackrel{}{=}  \E \big[  \lambda(a, a';X) \big |C=c]   \\	
	& \stackrel{}{\equiv} \rho(c, a, a'),
	\end{split}
\end{equation*}
which establishes that the center-specific average treatment effect is identified by $\rho(c, a, a').$
\end{proof}

\subsection{Identification under a weaker consistency assumption}\label{appendix:alternative_consistency}

Our results in the main text of the paper and thus far in this appendix have relied on condition A1 to connect the observed outcomes with potential (counterfactual) outcomes. Intuitively, condition A1 may be implausible if the centers in the trial apply different ancillary treatments or policies that may directly affect the outcome (not through treatment $A$). To address this possibility, we shall consider counterfactual outcomes under the joint intervention to implement all outcome-relevant policies of center $c$ \emph{and} set treatment $A$ to $a$; we denote these potential outcomes as $Y^{c,a}$. Now, consider the following consistency assumption: 

\vspace{0.05in}
\noindent
\emph{A1$^*$ – Consistency of potential outcomes for joint interventions on center-participation and treatment:} if $C_i=c$ and $A_i = a$, then $Y_i^{c,a}=Y_i$, for every individual $i$ in the population underlying the randomized trial, each treatment $a \in \mathcal A$, and each center $c \in \mathcal C$.

Condition A1$^*$ connects the potential outcomes $Y^{c,a}$ to the observed outcome $Y$, but is restricted to individuals with $C=c$ and $A = a$; thus, condition A1$^*$ is plausible even if different centers apply different ancillary treatments or policies that may directly affect the outcome. In fact, condition A1$^*$ even allows for direct effects of center participation on the outcome, for example, through Hawthorne-type effects \citelatex{dahabreh2019generalizing}. 

Note that we could obtain condition A1 from A1$^*$, if we assumed that $Y_i^{c,a}=Y_i^a$, for every individual $i$ in the population underlying the randomized trial, for each treatment $a \in \mathcal A$, and each center $c \in \mathcal C$. This additional condition is an exclusion restriction assumption that center participation does not have effects on the outcome except through the assigned treatment \citelatex{dahabreh2019generalizing}. 

We will now show that the center-specific expectation of the potential outcomes under joint intervention to implement all outcome-relevant policies of center $c$ \emph{and} set treatment $A$ to $a$, that is, $\E[Y^{c,a}|C=c]$, is identifiable by $\phi(c, a)$. To do so, we replace condition A2, with the following exchangeability condition: 

\vspace{0.05in}
\noindent
\emph{A2$^*$ – Conditional exchangeability over $A$ in the trial for joint interventions on center-participation and treatment}: for every $a \in \mathcal A$ and every center $c \in \mathcal C$, $Y^{c,a} \independent A |(X,C=c)$. 

\begin{proposition}
Under conditions A1$^*$, A2$^*$, and A3, \textup{$\E[Y^{c,a} |C=c]$} is identified by the observed data functional \textup{$\phi(c, a)=\E \big[\E[Y|X, C=c, A=a] \big|C=c \big]$}.
\end{proposition}

\begin{proof}
Starting from the causal quantity of interest,
\begin{equation}\label{id:phi_APP_alternative}
	\begin{split}
\E [Y^{c,a}|C=c ]	& \stackrel{}{=}  \E \big[\E[ Y^{c,a}| X, C=c]\big |C=c] \\
	& \stackrel{}{=}  \E \big[\E[ Y^{c,a}|X, C=c, A=a]\big |C=c] \\
	& \stackrel{}{=}   \E \big[ \E[ Y |X, C=c, A=a] |C=c  \big] \\	
	& \stackrel{}{\equiv} \phi(c, a),
	\end{split}
\end{equation}
where the first step follows by the law of iterated expectation, the second by condition A2$^*$, the third by condition A1$^*$, last by the definition of $\phi(c, a)$; and quantities conditional on $A = a$ are well-defined by condition A3. 
\end{proof}
This result shows that under conditions A1$^*$, A2$^*$, and A3 the observed data functional $\phi(c, a)$ can be interpreted as the expectation of the potential outcome under joint intervention to implement all outcome-relevant policies of center $c$ and set treatment $A$ to $a$, $\E [Y^{c,a}|C=c ]$. Interestingly, as shown in Appendix Section \ref{appendix:identification_POM}, under conditions A1, A2, and A3, the same functional can be interpreted as the expectation of the potential outcome under intervention to set treatment $A$ to $a$ regardless of center participation, that is $\E [Y^{a}|C=c ]$ .

\clearpage
\setcounter{figure}{0}
\setcounter{equation}{0}
\section{Relation to ``transportability analyses''} \label{appendix:connection_transportability1}


\subsection{Identifiability condition for transporting inferences across centers}

In this section, we will introduce an alternative causal condition that is often used in the literature on transportability analyses \citelatex{dahabreh2020toward, dahabreh2019efficient}:

\vspace{0.05in}
\noindent
\emph{A5 – Exchangeability across centers conditional on covariates}: $Y^a \independent C|X$. 

\noindent For notational convenience, for each center $c$ in the collection of centers $\mathcal C$ we define a new set, $$ \mathcal C_{(c)} = \{ j : j \in \mathcal C, j \neq c \};$$ informally, this is the collection of all centers except center $c$. Furthermore, we define a new random variable, $R_{(c)}$, as the indicator function for belonging in $\mathcal C_{(c)}$, $$ R_{(c)} = I(C \in \mathcal C_{(c)}). $$

We now explore some implications of the identifiability condition A5. By Lemma 4.2 of \citelatex{dawid1979conditional}, condition A5 implies exchangeability over $C$ among the collection of centers (excluding the target center), that is for every $a \in \mathcal A$ and every $c \in \mathcal C$,
\begin{equation}\label{eq:independencies1}
Y^a \independent C | X \implies  Y^a \independent C | (X, R_{(c)}=1).
\end{equation}
Furthermore, condition A5 also implies exchangeability of participants in the collection of centers (excluding the target center) and the target center,
\begin{equation}\label{eq:independencies2}
Y^a \independent C | X \implies Y^a \independent I(C=c) | X \Longleftrightarrow Y^a \independent R_{(c)} | X.
\end{equation}
The first of these results will be useful to derive restrictions on the law of the observed data; the second result will be useful in obtaining identification results for causal quantities in the target center using information from the collection of all other centers in the trial. 

By Lemma 4.2 of \citelatex{dawid1979conditional}, the result in \eqref{eq:independencies1} and condition A2 imply that, conditional on covariates, the potential outcomes $Y^a$ are independent of $(C, A)$ in the collection of centers:
\begin{equation}\label{eq:independencies3}
    \begin{Bmatrix}
Y^a \independent A | (X, R_{(c)}=1, C) \\ 
Y^a \independent C | (X, R_{(c)}=1) \\ 
\end{Bmatrix} \iff Y^a \independent (C, A) | (X, R_{(c)}=1).
\end{equation}
Noting that $Y^a \independent (C, A) | (X, R_{(c)}=1)$ implies $Y^a \independent C  | (X, R_{(c)}=1, A = a)$, by the consistency condition A1, we obtain
\begin{equation}\label{eq:independencies4}
    Y \independent C | (X, R_{(c)}=1, A = a).
\end{equation}
Thus, among individuals participating in any center, excluding the target center $c$, that is to say, when $R_{(c)}=1$, the observed outcome $Y$ is independent of center membership $C$, within treatment groups ($A = a$) and conditional on covariates $X$. Because this condition does not involve potential outcomes, it is testable using the observed data (e.g., using methods for comparing conditional densities). Furthermore, because $\{R_{(c)}=0\} \Longleftrightarrow \{R_{(c)}=0,C=c\}$, we also obtain that $Y \independent C | (X, R_{(c)}, A = a)$. 

The independence condition in \eqref{eq:independencies4} implies that for every $a \in \mathcal A$ and every $x$ such that $f(x, C = c) \neq 0$,
\begin{equation}\label{eq:observed_data_implications2}
  \E[Y | X = x, C = 1, A = a] = \ldots = \E[Y | X = x, C = m, A = a],
\end{equation} 
which is also testable using the observed data.

\subsection{Identification of center-specific effects using ``transportability analyses''}

In this section, we will obtain an alternative identification result for $\E[Y^a |C=c]$ by replacing condition A4 with the transportability condition A5.

\begin{proposition}
Under conditions A1 through A3, and A5, \textup{$\E[Y^a |C=c]$} is identified by the observed data functional \textup{$\chi(c, a) \equiv  \E \big[ \E[ Y |X, R_{(c)}=1, A=a] |C=c  \big]$}, for each treatment $a \in \mathcal A$ and each center $c \in \mathcal C$.
\end{proposition}

\begin{proof}
Starting from the causal quantity of interest,
\begin{equation}\label{id:chi_APP}
	\begin{split}
\E [Y^a|C=c ]	& \stackrel{}{=}  \E \big[\E[ Y^a| X, C=c]\big |C=c] \\
	& \stackrel{}{=}  \E \big[\E[ Y^a|X, C=c, A=a]\big |C=c] \\
	& \stackrel{}{=}   \E \big[ \E[ Y |X, C=c, A=a] |C=c  \big] \\	
	& \stackrel{}{=}   \E \big[ \E[ Y |X, R_{(c)}=1, A=a] |C=c  \big] \\
	& \stackrel{}{\equiv} \chi(c, a),
	\end{split}
\end{equation}
where the first step follows by the law of iterated expectation, the second by condition A2, the third by condition A1, the fourth by condition A5 (and the lemmas introduced in the previous subsection, \ref{eq:independencies4}), the last by the definition of $\chi(c, a)$, and quantities conditional on $A = a$ are well-defined by condition A3. 
\end{proof}
The inner expectation in this expression pools information across all centers in the trial \emph{except} for the target center of interest in the trial $C=c$, among individuals assigned to treatment group $A=a$. In other words, $ \chi(c, a)$ is transporting inference from all centers in the trial except the target center, to the target center. This identifiability result is closely related to results on transporting inferences about causal effects from a trial (here, the aggregation of all centers in $\mathcal C_{(c)}$) to a target population (here, the population represented by center $c$) \citelatex{dahabreh2020extending}. In fact, this result underpins some recent analyses of multicenter trial data \citelatex{rudolph2018composition, rudolph2017robust}. In \ref{appendix:connection_transportability} we examine the estimation strategies suggested by applying this result to multicenter trials compared to the strategies we proposed in the main text of our paper (as well as \ref{appendix:influence_function}). We show that when both approaches are valid, the approach suggested in the main text of our paper has the potential to produce more efficient estimators of potential outcome mean compared to ``transportability'' analyses.

\subsection{Deriving condition A4 from causal conditions}

In this section, we show that condition A4 can be derived from the more primitive (elementary) causal assumption A5 that is often used in the literature on transportability analyses \citelatex{dahabreh2020toward, dahabreh2019efficient}. By Lemma 4.2 of \citelatex{dawid1979conditional}, conditions A2 ($ Y^a \independent A|X,C $) and A5 ($Y^a \independent C|X$) imply that $Y^a \independent (C,A)|X$. By the weak union property of conditional independence, $Y^a \independent (C,A)|X$ implies that $Y^a \independent C|(X,A)$, which in turn, by consistency, implies that $Y \independent C|(X,A = a)$, for every $a \in \mathcal A$. This last independence condition is precisely condition A4 (i.e., the condition of no center-outcome associations given covariates and treatment). Therefore, if we believe condition A5 holds, given that condition A2 is expected to hold by design, we should also believe that condition A4 holds (recall that condition A4 underpins the identification results about $\psi(a,c)$ given earlier in this Appendix and in the main text of the paper). The converse, however is not true: condition A4 can hold even if condition A5 does not.

\clearpage
\setcounter{figure}{0}
\section{Influence functions and estimators}\label{appendix:influence_function}
In this section, we focus on presenting results on the influence functions and estimators for $\phi(c, a)$ and $\psi(c,a)$, under nonparametric and semiparametric models. In the last subsection, we build on the results for potential outcome means to derive influence functions and estimators for center-specific average treatment effects. Throughout, we work under the biased sampling model, described in \ref{appendix:sampling_model} \citelatex{breslow2000semi}; see references for similar arguments in different settings \citelatex{dahabreh2019efficient, dahabreh2020extending, kennedy2015semiparametric}.

\subsection{Influence functions under the nonparametric model}

We now obtain the first-order influence function \citelatex{bickel1998efficient} for each functional $\phi(c, a)$ and $\psi(c, a)$ under the nonparametric model $\mathcal M_{\text{\tiny np}}$ for the law of the observed data, $O = (X, C, A, Y)$. 

\subsubsection{Center-outcome associations may be present}

Under $\mathcal M_{\text{\tiny np}}$, the influence function for $\phi(c, a)$ is 
\begin{equation}\label{eq:inf_func_phi}
    \begin{split}
\mathit{\Phi}^1_{p_0}(c, a) &= \dfrac{1}{\Pr_{p_0}[C=c]} \Bigg\{ \dfrac{I(C = c, A = a)}{\Pr_{p_0}[A = a| X, C=c]} \big\{ Y - \E_{p_0}[Y | X, C = c, A = a]  \big\} \\ 
&\quad\quad\quad\quad\quad\quad\quad\quad\quad\quad\quad\quad\quad\quad\quad\quad + I(C=c)\big\{ \E_{p_0}[Y|X, C = c , A = a] - \phi_{p_0}(c, a) \big\} \Bigg\},
    \end{split}
\end{equation}
where the subscript $p_0$ denotes that all quantities are evaluated at the ``true'' data law. Specifically, $\mathit\Phi^1_{p_0}(c, a)$ satisfies
\[
\dfrac{\partial \phi_{p_t}(c,a)}{\partial t} \Bigg|_{t=0}  = \E[\mathit\Phi^1_{p_0}(c, a) u(O)],
\] where $u(O)$ denotes the score of the observed data and the left hand side of the above equation is the pathwise derivative of the target parameter $\phi(c, a)$. Theorem 4.4 in \citelatex{tsiatis2006} shows that $\mathit\Phi_{p_0}^1(c, a)$ lies in the tangent set; it follows (see, e.g., \citelatex{van2000asymptotic}, page 363) that $\mathit\Phi^1_{p_0}(c, a)$ is the unique (and efficient) influence function under the nonparametric model for the observed data. 

\subsubsection{Center-outcome associations are absent}

The influence function for $\psi(c, a)$ is
\begin{equation}\label{eq:inf_func_psi}
    \begin{split}
\mathit{\Psi}^1_{p_0}(c, a) &= \dfrac{1}{\Pr_{p_0}[C=c]} \Bigg\{ \dfrac{I(A = a) \Pr_{p_0}[C = c | X]}{\Pr_{p_0}[A = a| X]} \big\{ Y - \E_{p_0}[Y | X, A = a]  \big\} \\ 
&\quad\quad\quad\quad\quad\quad\quad\quad\quad\quad\quad\quad\quad + I(C=c)\big\{ \E_{p_0}[Y|X, A = a] - \psi_{p_0}(c, a) \big\} \Bigg\}.
    \end{split}
\end{equation}
Following similar arguments to those for $\mathit\Phi^1_{p_0}(c, a)$ and $\phi(c,a)$, we can conclude that $\mathit\Psi^1_{p_0}(c, a)$ is the unique (efficient) influence function for $\psi(c, a)$ under the nonparametric model for the observed data. 

\subsection{Influence function under semiparametric models}

\subsubsection{Center-outcome associations may be present}

We will now argue that $\mathit\Phi^1_{p_0}(c, a)$ is the efficient influence function under a semiparametric model $\mathcal M_{\text{\tiny semi}}$ where the probability of treatment conditional on covariates and center participation, $p(a|x,c)$, is known. Under this model, we can decompose the tangent space as $\Lambda_{\text{\tiny semi}} =  \Lambda_{C} \oplus  \Lambda_{X|C} \oplus  \Lambda_{Y|X,C, A}$ \citelatex{tsiatis2006}. 
Rewriting the influence function under the nonparametric model as the sum of two terms, we obtain
\begin{equation}\label{eq:inf_func_rewritephi}
  \begin{split}
    \mathit\Phi^1_{p_0}(c, a) &=  \dfrac{I(C=c,A = a)}{\Pr_{p_0}[C = c]\Pr_{p_0}[A = a| X,C] }\Big\{ Y - \E_{p_0}[Y | X, C=c, A = a]\Big\}  \\
      &\quad\quad\quad +  \dfrac{I(C=c)}{\Pr_{p_0}[C = c]} \Big\{ \E_{p_0}[Y | X, C=c, A = a] - \phi_{ p_0}(c, a)\Big\}.
  \end{split}
\end{equation}
The first term is a function of $(X, C, A, Y)$ that has mean zero conditional on $(X,C,A)$, and thus belongs to $\Lambda_{Y|X,C, A}$. Furthermore, the second term in the above expression is a function of $(X,C)$, with mean zero conditional on $C$, and thus belongs to $\Lambda_{X|C}$. From these observations we conclude that the influence function under the semiparametric model, $\mathit\Phi^1_{p_0}(c, a)$, belongs to $\Lambda_{\text{\tiny semi}}$ and its projection onto that space is equal to itself. We can conclude that the unique influence function under the nonparametric model, $\mathit\Phi^1_{p_0}(c, a)$, is also the efficient influence function under the semiparametric model $\mathcal M_{\text{\tiny semi}}$.

\subsubsection{Center-outcome associations are absent}

When center-outcome association are absent, we need to examine the implications of two kinds of restrictions on the law of the observed data. First, we need to examine the impact of the independence condition $Y \independent C | (X, A)$. Second, we need to examine the impact of knowing the probability of treatment, that is, of knowing the conditional density $p(a|x,c)$.

\paragraph{Incorporating the restriction of no center-outcome associations, given covariates and treatment:}

Consider the semiparametric model $\mathcal M_{\text{\tiny semi}}^*$ that incorporates the restriction $Y\independent C | (X, A)$, that is $p(y|x,c,a) = p(y|x,a)$. Under this model, the density of the law of the observed data is 
\begin{equation*}
    p(x,c,a,y) = p(c) p(x|c) p(a|x,c) p(y|x, a),
\end{equation*}
and we obtain the tangent space decomposition $\Lambda_{\text{\tiny semi}}^* =  \Lambda_{C} \oplus  \Lambda_{X|C} \oplus \Lambda_{A|X, C} \oplus \Lambda_{Y|X, A}$. 
Rewriting the influence function under the nonparametric model as the sum of two terms, we obtain
\begin{equation}\label{eq:inf_func_rewrite}
  \begin{split}
    \mathit\Psi^1_{p_0}(c,a) &=  \dfrac{I(A = a)\Pr_{p_0}[C = c | X]}{\Pr_{p_0}[C = c]\Pr_{p_0}[A = a| X] }\Big\{ Y - \E_{p_0}[Y | X, A = a]\Big\}  \\
      &\quad\quad\quad +  \dfrac{I(C=c)}{\Pr_{p_0}[C = c]} \Big\{ \E_{p_0}[Y | X, A = a] - \psi_{p_0}(c, a)\Big\}.
  \end{split}
\end{equation}
The first term is a function of $(X, A, Y)$ that has mean zero conditional on $(X, A)$, and thus belongs to $\Lambda_{Y|X, A}$. Furthermore, the second term in the above expression is a function of $(X,C)$, with mean zero conditional on $C$, and thus belongs to $\Lambda_{X|C}$. From these observations we conclude that the influence function under the semiparametric model, $\mathit\Psi^1_{p_0}(c, a)$, belongs to $\Lambda_{\text{\tiny semi}}^*$ and its projection onto that space is equal to itself. We can conclude that the unique influence function under the nonparametric model, $\mathit\Psi^1_{p_0}(c, a)$, is also the efficient influence function under the semiparametric model $\mathcal M_{\text{\tiny semi}}^{*}$.

\paragraph{Knowing the probability of treatment:}

Consider now the semiparametric model $\mathcal M_{\text{\tiny semi}}^{**}$ for the law of the observed data where the restriction $Y\independent C | (X, A)$ holds and $p(a|x,c)$ is known. We can decompose the tangent space as $\Lambda_{\text{\tiny semi}}^{**} =  \Lambda_{C} \oplus  \Lambda_{X| C} \oplus \Lambda_{Y|X, A}$. Considering the expression for the influence function under the nonparametric model in \eqref{eq:inf_func_rewrite}, we see that the first term is a function of $(X, A, Y)$ that has mean zero conditional on $(X,A)$ thus it belongs to $\Lambda_{Y|X, A}$. Furthermore, the second term in the above expression is a function of $(X,C)$ that has mean zero conditional on $C$, thus it belongs to $\Lambda_{X|C}$. Thus, the influence function under the nonparametric model, $\mathit\Psi^1_{p_0}(c, a)$, belongs to $\Lambda_{\text{\tiny semi}}^{**}$ and its projection onto that space is equal to itself. We can conclude that the unique influence function under the nonparametric model, $\mathit\Psi^1_{p_0}(c, a)$, is also the efficient influence function under the semiparametric model $\mathcal M_{\text{\tiny semi}}^{**}$.

\subsection{Proposed estimators}

\subsubsection{Center-outcome associations may be present}
The influence function in display \eqref{eq:inf_func_phi} suggests that the following is a reasonable estimator for $\phi(c, a)$: \begin{equation}\label{eq:estimatorphi2_APP} 
    \begin{split} \widehat \phi(c, a) &= \dfrac{1}{ n_c} \sum\limits_{i=1}^{n} \Bigg\{ \dfrac{I(C_i=c,A_i = a)}{ \widehat e_a(X_i,C_i)} \Big\{ Y_i - \widehat g_a(X_i,C_i) \Big\} + I(C_i = c) \widehat g_a(X_i,C_i) \Bigg\}, 
    \end{split}
\end{equation} 
where $n_c = \sum_{i=1}^{n} I(C_i=c)$; $\widehat g_a(X,C)$ is an estimator for $\E[Y |X, C, A = a]$; and $\widehat e_a(X,C)$ is an estimator for $\Pr[A = a| X, C]$.

The adjusted estimator above is usually more efficient than the common approach that uses a crude analysis, where the marginal expectation of each treatment group is calculated within a center,
\begin{equation}\label{eq:tau_APP} 
    \begin{split} \widehat \tau(c, a) &= \dfrac{1}{ n _{c,a} } \sum\limits_{i=1}^{n} I(C_i=c, A_i=a) Y_i , 
    \end{split}
\end{equation} 
where $n_{c,a} = \sum_{i=1}^{n} I(C_i = c, A_i = a)$.

\subsubsection{Center-outcome associations are absent}

Suppose now that background knowledge suggests that the outcome is independent of center, conditional on covariates and treatment (i.e., assumption A4 is plausible) and the data does not provide evidence against it, such that the identification result in display \eqref{id:psi_APP} is valid, and $\psi(c, a)$ identifies the center-specific treatment effects. The influence function in display \eqref{eq:inf_func_psi} suggests that the following is a reasonable estimator  $\psi(c, a)$:
\begin{equation}\label{eq:estimatorpsi_APP}
    \begin{split}
  \widehat \psi(c, a) &= \dfrac{1}{ n_c } \sum\limits_{i=1}^{n} \Bigg\{ I(A_i = a) \dfrac{\widehat p_{c}(X_i)}{ \widetilde e_a(X_i)}  \Big\{ Y_i - \widetilde g_a(X_i) \Big\} + I(C_i = c) \widetilde g_a(X_i) \Bigg\},
  \end{split}
\end{equation}
where $n_c = \sum_{i=1}^{n} I(C_i=c)$; $\widetilde g_a(X)$ is an estimator for $\E[Y | X,A = a]$; $\widehat p_{c}(X)$ is an estimator for $\Pr[C = c| X]$; and $\widetilde e_a(X)$ is an estimator for $\Pr[A = a| X]$.

\subsection{Influence functions and estimators for the average treatment effect}
Here we present results on the influence functions and estimators for $\delta_{\phi}(c,a,a^\prime)$ and $\delta_{\psi}(c,a,a^\prime)$, using the results for potential outcomes means. 

\subsubsection{Center-outcome associations may be present}
The influence function for $\delta_{\phi}(c,a,a^\prime)$ is
$$\mathit{\Delta}^{1}_{\phi, p_0}(c,a,a^\prime)=\mathit{\Phi}^{1}_{p_0}(c,a) - \mathit{\Phi}^{1}_{p_0}(c,a^\prime).$$
An estimator for $\delta_{\phi}(c,a,a^\prime)$ can be obtained as $\widehat \delta_{\phi}(c,a,a^\prime) = \widehat \phi(c,a) - \widehat \phi(c,a^\prime)$.

\subsubsection{Center-outcome associations may be absent}
The influence function for $\delta_{\psi}(c,a,a^\prime)$ is
$$\mathit{\Delta}^{1}_{\psi, p_0}(c,a,a^\prime)=\mathit{\Psi}^{1}_{p_0}(c,a) - \mathit{\Psi}^{1}_{p_0}(c,a^\prime).$$
An estimator for $\delta_{\psi}(c,a,a^\prime)$ can be obtained as $\widehat \delta_{\psi}(c,a,a^\prime) = \widehat \psi(c,a) -  \widehat \psi(c,a^\prime)$.

\clearpage
\section{Estimation by weighting or g-computation} \label{appendix:estimation_gcomp_weighting} 
\setcounter{figure}{0}

To gain intuition into the doubly robust estimators, $\widehat \phi(c,a)$ and  $\widehat \psi(c,a)$, from the main text, it is useful to consider (non-robust) weighting \citelatex{rosenbaum1983, robins1992} or g-computation (outcome-model based) estimators \citelatex{greenland1986, robins1992}.

\subsection{Center-outcome associations may be present}

Consider $\widehat \phi(c,a)$. A weighting estimator is obtained by setting the weights in $\widehat \phi(c,a)$ to zero, 
\begin{equation*}\label{eq:estimatorphi_ipwAPP}
    \begin{split}
  \widehat \phi_{\text{\tiny{IPW}}}(c, a) &= \dfrac{1}{ n_c } \sum\limits_{i=1}^{n}  I(C_i=c, A_i = a) \dfrac{1}{ \widehat e_a(X_i, C_i)} Y_i .
  \end{split}
\end{equation*}
An outcome modeling estimator is obtained by setting the weights in $\widehat \phi(c, a)$ to zero, 
\begin{equation*}\label{eq:estimatorphi_omAPP}
    \begin{split}
  \widehat \phi_{\text{\tiny{OM}}}(c, a) &= \dfrac{1}{ n_c } \sum\limits_{i=1}^{n}  I(C_i = c) \widehat g_a(X_i, C_i).
  \end{split}
\end{equation*}

\subsection{Center-outcome associations are absent}

Now, consider $\widehat \psi(c,a)$ from the main text. A weighting estimator is obtained by setting $\widehat g_a(X)$ to zero, 
\begin{equation*}\label{eq:estimatorpsi_ipwAPP}
    \begin{split}
  \widehat \psi_{\text{\tiny{IPW}}}(c, a) &= \dfrac{1}{ n_c } \sum\limits_{i=1}^{n}  I(A_i = a) \dfrac{\widehat p_{c}(X_i)}{ \widetilde e_a(X_i)} Y_i .
  \end{split}
\end{equation*}
An outcome modeling estimator is obtained by setting the weights in $\widehat \psi(c, a)$ to zero, 
\begin{equation*}\label{eq:estimatorpsi_omAPP}
    \begin{split}
  \widehat \psi_{\text{\tiny{OM}}}(c, a) &= \dfrac{1}{ n_c } \sum\limits_{i=1}^{n}  I(C_i = c) \widetilde g_a(X_i).
  \end{split}
\end{equation*}

As in the previous sections, estimators of center-specific average treatment effects can be obtained by taking differences between the center-specific potential outcome mean estimators.

\clearpage
\section{Asymptotic properties of robust estimators} \label{appendix:asymptotic} 
\setcounter{figure}{0}

\subsection{Model robustness}

\subsubsection{Robustness for analyses when center-outcome associations may be present:}

Suppose that $\widehat e_{a}(X,C)$ and $\widehat g_{a}(X,C)$ converge in probability to well-defined limits that we denote as $e_{a}^*(X,C)$ and $g_{a}^*(X,C)$, respectively. Then, as $n \longrightarrow \infty$,
\begin{equation*}
    \begin{split}
        \widehat \phi(c, a) & \overset{p}{\longrightarrow} \dfrac{1}{\Pr[C = c]} \E \Bigg[ \dfrac{I(C = c, A = a)}{e_{a}^*(X, C)} \big\{ Y - g_{a}^*(X, C) \big\}  + I(C=c) g_{a}^*(X, C)  \Bigg].
    \end{split}
\end{equation*}

\noindent
Because the model for the treatment is known or if estimated cannot be misspecified, $e_{a}^*(X, C) = \Pr[A = a | X, C]$ always. However, the outcome model may be misspecified, so $g_{a}^*(X, C)$ is not necessarily equal to $\E[Y | X, C, A = a]$. Here we show that $\widehat \phi(c, a)$ remains robust even if the outcome model is misspecified.
To show this we can re-write the above expression as three terms inside the curly brackets and consider each term separately:
\begin{equation*}
    \begin{split}
        \widehat \phi(c, a) & \overset{p}{\longrightarrow} \dfrac{1}{\Pr[C = c]} \Bigg\{ \E \Bigg[ \dfrac{I(C = c, A = a)}{e_{a}^*(X, C)}Y\Bigg] - \E \Bigg[ \dfrac{I(C = c, A = a)}{e_{a}^*(X, C)}g_{a}^*(X, C) \Bigg] \\ &\quad\quad\quad+ \E \Bigg[I(C=c) g_{a}^*(X, C)  \Bigg]\Bigg\} .
    \end{split}
\end{equation*}

Consider the first term:
\begin{equation*}
	\begin{split}
\E \left[ \dfrac{I(C = c, A = a)}{e_{a}^*(X, C)}Y\right]
   & \stackrel{}{=} \E \Bigg[  \E \left[ \dfrac{I(C = c, A = a)}{e_{a}^*(X, C)}Y | X \right] \Bigg]  \\
	& \stackrel{}{=} \E \big[ \E[Y|X, C=c, A=a] \Pr[C=c|X] \big]  \\
	& \stackrel{}{=}  \E \big[ \E[Y|X, C=c, A=a] |C=c \big] \Pr[C=c]  \\	
	& \stackrel{}{=} \phi(c, a) \Pr[C=c].
	\end{split}
\end{equation*}

Now, consider the second term:
\begin{equation*}
	\begin{split}
\E \left[ \dfrac{I(C = c, A = a)}{e_{a}^*(X, C)}g_{a}^*(X, C) \right]
   & \stackrel{}{=} \E \Bigg[ \E \left[ \dfrac{I(C = c, A = a)}{e_{a}^*(X, C)}g_{a}^*(X, C) |X \right] \Bigg]  \\
	& \stackrel{}{=} \E \big[  g_{a}^*(X, C) \Pr[C=c|X] \big]  \\
	& \stackrel{}{=} \E \big[ I(C=c) g_{a}^*(X, C) \big]  \\
	\end{split}
\end{equation*}

Combining the above results, we have $\widehat \phi(c, a) \overset{p}{\longrightarrow}  \phi(c, a) $ even when $g_{a}^*(X, C)  \neq \E[Y | X,C=c,A=a]$, and we can conclude that analyses are robust to outcome model misspecification. 

\subsubsection{Double robustness when center-outcome associations are absent:}

Suppose that $\widehat e_a(X)$ and $\widehat p_{c}(X)$ converge in probability to well-defined limits that we denote as $\widehat e^{*}_a(X)$ and $\widehat p^{*}_{c}(X)$, respectively. Then as
$n \longrightarrow \infty$, 
\begin{equation}\label{eq:appendix:estimator_limit}
  \begin{split}
\widehat \psi(c, a) &\overset{p}{\longrightarrow} \dfrac{1}{\Pr[C = c]} \Bigg\{ \E \left[I(A = a)  \dfrac{p^{*}_{c}(X)}{ e^*_a(X)}  \Big\{ Y - g_a^*(X) \Big\} \right]   \\
&\quad\quad\quad+  \E \left[ I(C = c) g_a^*(X) \right]\Bigg\}. 
  \end{split}
\end{equation}

\vspace{0.1in}
\noindent
\emph{Case 1: correct specification of the model for the probability of participation in the target center and the treatment assignment model.} Assume that the models for $\Pr[C = c| X ]$ and $\Pr[A = a | X]$ are correctly specified, such that
\begin{align*}
\widehat p_{c}(X) &\overset{p}{\longrightarrow} p^{*}_{c}(X) = \Pr[ C = c | X ] \\
\widehat e_a(X) &\overset{p}{\longrightarrow} e^*_a(X) = \Pr[A = a | X].
\end{align*}
We do not, however, assume that the asymptotic limit $g_{a}^{*}(X)$ is equal to the corresponding target parameter.
To show this we can re-write the above expression as three terms inside the curly brackets and consider each term separately:
\begin{equation*}
    \begin{split}
        \widehat \psi(c, a) & \overset{p}{\longrightarrow} \dfrac{1}{\Pr[C = c]} \Bigg\{ \E \Bigg[ I(A = a)  \dfrac{p^{*}_{c}(X)}{ e^*_a(X)} Y\Bigg] - \E \Bigg[ I(A = a)  \dfrac{p^{*}_{c}(X)}{ e^*_a(X)}g_{a}^*(X) \Bigg] \\ &\quad\quad\quad+ \E \Bigg[I(C=c) g_{a}^*(X)  \Bigg]\Bigg\} .
    \end{split}
\end{equation*}

Consider the first term:

\begin{equation*}
	\begin{split}
\E \left[ I(A = a)  \dfrac{p^{*}_{c}(X)}{ e^*_a(X)}Y \right]
   & \stackrel{}{=} \E \Bigg[  \E \left[I(A = a)  \dfrac{p^{*}_{c}(X)}{ e^*_a(X)}Y | X \right] \Bigg]  \\
	& \stackrel{}{=} \E \big[ \E[Y|X, A=a] \Pr[C=c|X] \big]  \\
	& \stackrel{}{=}  \E \big[ \E[Y|X,A=a] |C=c \big] \Pr[C=c]  \\	
	& \stackrel{}{=} \psi(c, a) \Pr[C=c].
	\end{split}
\end{equation*}

Now, consider the second term:

\begin{equation*}
	\begin{split}
\E \left[ I(A = a)  \dfrac{p^{*}_{c}(X)}{ e^*_a(X)}g_{a}^*(X) \right]
   & \stackrel{}{=} \E \Bigg[ \E \left[ I(A = a)  \dfrac{p^{*}_{c}(X)}{ e^*_a(X)}g_{a}^*(X) |X \right] \Bigg]  \\
	& \stackrel{}{=} \E \big[  g_{a}^*(X) \Pr[C=c|X] \big]  \\
	& \stackrel{}{=} \E \big[ I(C=c) g_{a}^*(X) \big]  \\
	\end{split}
\end{equation*}

Combining the above results, we have $\widehat \psi(c, a) \overset{p}{\longrightarrow}  \psi(c, a) $ even when $g_{a}^*(X)  \neq \E[Y | X,A=a]$, and we can conclude that analyses are robust to outcome model misspecification.

\vspace{0.1in}
\noindent
\emph{Case 2: correct specification of the model for the outcome mean.} Assume that the model for $\E[Y|X, A = a]$ is correctly specified, such that
\begin{equation*}
\widehat g_{a}(X) \overset{p}{\longrightarrow} g_{a}^{*}(X) =\E[Y|X, A = a].
\end{equation*}
We do not, however, assume that the asymptotic limit $p^{*}_{c}(X)$ equals $\Pr[C =c | X ]$ or that the asymptotic limit $e_a^*(X)$ equals $\Pr[A = a | X]$. Using iterated expectation arguments for the first time on the right-hand-side of (\ref{eq:appendix:estimator_limit}) we obtain 
\begin{equation*}
\begin{split}
& \E \left[I(A = a)  \dfrac{p^{*}_{c}(X)}{e_a^*(X)}  \Big\{ Y - g_a^*(X) \Big\} \right] = 0.
\end{split}
\end{equation*}
and $\E[I(C=c)g_a^*(X)]=\Pr[C=c]\psi(c, a)$, so
$\widehat \psi(c, a) \overset{p}{\longrightarrow} \psi(c, a)$. Cases 1 and 2 establish the double robustness of $\widehat \psi(c, a)$.

\subsection{Asymptotic efficiency}
In this section of the Appendix we discuss how estimators of the quantities $\phi(c, a)$ and $\psi(c, a)$ compare in terms of efficiency. Any regular estimator of the functionals $\phi(c, a)$ and $\psi(c, a)$ has an asymptotic variance that is less than or equal to $\E\left[ \big( \mathit{\Phi}^1_{p_0}(c, a) \big)^2 \right]$ or $\E\left[ \big( \mathit{\Psi}^1_{p_0}(c, a) \big)^2  \right]$, respectively. Using this definition, the variance bound for $\phi(c, a)$ is 
\begin{equation}\label{eq:var_bound_phi}
  \begin{split}
  &\E\left[ \big( \mathit{\Phi}^1_{p_0}(c, a) \big)^2 \right] = \frac{1}{ (\Pr_{p_0}[C=c]){^2}}  \Bigg\{ \E \left[ \dfrac{I(C = c, A = a)}{\big( \Pr_{p_0}[A = a | X, C = c] \big)^2} \big\{ Y - \E_{p_0}[Y | X, C = c, A = a] \big\}^2 \right] \\
  &\quad\quad\quad\quad\quad\quad\quad\quad\quad\quad\quad\quad\quad + \E \left[ I(C = c) \big\{ \E_{p_0}[Y | X, C = c, A = a] - \phi_{p_0}(c, a) \big\}^2 \right] \Bigg \}.
  \end{split}
\end{equation}

And the variance bound for $\psi(c, a)$ is
\begin{equation}\label{eq:var_bound_psi}
  \begin{split}
  &\E\left[ \big( \mathit{\Psi}^1_{p_0}(c, a) \big)^2  \right] = \frac{1}{(\Pr_{p_0}[C=c]){^2}}  \Bigg\{ \E \left[ \dfrac{I(A = a)  \big(\Pr_{p_0}[C = c | X]\big)^2}{ \big( \Pr_{p_0}[A = a | X]\big)^2} \big\{ Y - \E_{p_0}[Y | X, A = a] \big\}^2 \right] \\
  &\quad\quad\quad\quad\quad\quad\quad\quad\quad\quad\quad + \E \left[ I(C = c) \big\{ \E_{p_0}[Y | X, A = a] - \psi_{p_0}(c, a) \big\}^2 \right] \Bigg \}.
  \end{split}
\end{equation}
It is interesting to compare the two bounds, because they represent a best-case comparison: provided we have estimators that are in the semiparametric efficiency sense optimal, then the bounds show the limit of their performance. In the remainder of this section, we omit the subscript ${p_0}$ for notational convenience.

When conditions A1 through A4 hold, $\phi(c, a)  = \psi(c, a)$ and $\E[Y | X, C = c, A = a]=\E[Y | X, A = a].$
Thus the second term inside the large brackets each of the expressions (\ref{eq:var_bound_phi}) and (\ref{eq:var_bound_psi}) is the same, 
\begin{equation*}
    \begin{split}
    &\E \left[ I(C = c) \big\{ \E[Y | X, C = c, A = a] - \phi(c, a) \big\}^2 \right]   \\ 
   &\quad\quad = \E \left[ I(C = c) \big\{ \E[Y | X, A = a] - \psi(c, a) \big\}^2 \right].
    \end{split}
\end{equation*}
We only need to compare the first term in the large brackets in expressions (\ref{eq:var_bound_phi}) and (\ref{eq:var_bound_psi}). An iterated expectation argument shows that the first such term, in expression (\ref{eq:var_bound_phi}) is equal to 
\begin{equation} \label{eq:first_term_var_bound_phi}
    \begin{split}    
&\E\left[ \dfrac{\E\big[ \big\{Y - \E[Y | X, C = c, A = a] \big\}^2 | X, C = c, A = a \big] \Pr[C = c | X]}{\Pr[A = a| X, C = c] } \right] \\
&\quad\quad\quad\quad\quad= \E\left[ \dfrac{\mbox{Var}[Y | X, C = c, A = a] \Pr[C = c | X]}{\Pr[A = a| X, C = c] } \right].
    \end{split}
\end{equation}
Similarly, the first term in in expression (\ref{eq:var_bound_psi}) is
\begin{equation}\label{eq:first_term_var_bound_psi}
    \begin{split}    
&\E\left[ \dfrac{\E\big[ \big\{Y - \E[Y | X, A = a] \big\}^2 | X, A = a\big] \big( \Pr[C = c | X]\big)^2}{\Pr[A = a| X] } \right] \\
&\quad\quad\quad\quad\quad= \E\left[ \dfrac{\mbox{Var}[Y | X, A = a] \big( \Pr[C = c | X]\big)^2}{\Pr[A = a| X] } \right].
    \end{split}
\end{equation}
By condition A4, however, we have that $$\mbox{Var}[Y | X, C = c, A = a]=\mbox{Var}[Y | X, A = a].$$ Furthermore, by the law of total probability, 
\begin{equation*}
    \begin{split}
    \Pr[A = a | X] &= {\Pr}[A =a | X, R_{(c)}=1] {\Pr}[R_{(c)}=1 | X] + {\Pr}[A = a | X, R_{(c)}=0] {\Pr}[ R_{(c)}=0| X] \\
  &= {\Pr}[A =a | X, R_{(c)}=1] {\Pr}[R_{(c)}=1 | X] + {\Pr}[A = a | X, C=c] {\Pr}[ C=c| X]. 
  \end{split}
\end{equation*}

These two facts and straightforward calculations, establish that 
$$ \E\left[ \dfrac{\mbox{Var}[Y | X, A = a] \big( \Pr[C = c | X]\big)^2}{\Pr[A = a| X] } \right] \leq \E\left[ \dfrac{\mbox{Var}[Y | X, C = c, A = a] \Pr[C = c | X]}{\Pr[A = a| X, C = c] } \right]. $$
It follows that $\E\left[ \big( \mathit{\Psi}^1_{p_0}(c, a) \big)^2  \right] \leq \E\left[ \big( \mathit{\Phi}^1_{p_0}(c, a) \big)^2 \right].$

\section{Inference}\label{appendix:inference}
To construct confidence intervals for the estimators $\widehat \phi(c, a)$ and $\widehat \psi(c, a)$, and the corresponding average treatment effect estimators, we can use the influence curve-based approximation of the standard error \citelatex{vanderLaan2003}.

\subsection{Center-outcome associations are present}
When center-outcome associations may be present, we can estimate the variance of the sampling distribution of $\widehat \phi(c,a)$ as $$\widehat \sigma^{2}_{\widehat \phi(c,a)}= \frac{1}{n}\widehat {\text{Var}}\left[\mathit{\widehat\Phi}^{1}(c,a)\right], $$ where $\widehat {\text{Var}}\left[ \widehat{\mathit{\Phi}}^{1}(c,a)\right]$ is the estimated variance of the influence curve, $\mathit{\widehat\Phi}^{1}(c,a)$, the sample analog of the influence function $\mathit{\Phi}^{1}(c,a)$. We can obtain a $(1- \alpha)\%$ confidence interval as $(\widehat \phi(c,a) \pm z_{1-\alpha/2} \times \widehat \sigma_{\widehat \phi(c,a)})$, where $z_{1-\alpha/2}$ is the $1- \alpha/2$ quantile of the normal distribution. 

Similarly, we can obtain the influence curve-based estimator of the standard error for the average treatment effect for $\widehat \delta_{\phi}(c,a,a^\prime)$. We can estimate the variance of the sampling distribution of $\widehat \delta_{\phi}(c,a,a^\prime)$ as $$\widehat \sigma^{2}_{\widehat \delta_{\phi}(c,a,a^\prime)}= \frac{1}{n}\widehat {\text{Var}}\left[\widehat{\mathit{\Delta}}^{1}_{\phi}(c,a, a^\prime)\right],$$
where $\widehat {\text{Var}}\left[\mathit{\widehat\Delta}^{1}_{\phi}(c,a, a^\prime)\right]$ is the estimated variance of the influence curve for the estimator of the center-specific average treatment effect. We can obtain a $(1- \alpha)\%$ confidence interval as $(\widehat \delta_{\phi}(c,a,a^\prime) \pm z_{1-\alpha/2} \times \widehat \sigma_{\widehat \delta_{\phi}(c,a,a^\prime)})$.
 
\subsection{Center-outcome associations are absent}
Similarly, when center-outcome associations are absent, we can estimate the variance of the sampling distribution of $\widehat \psi(c,a)$ as $$\widehat \sigma^{2}_{\widehat \psi(c,a)}= \frac{1}{n}\widehat {\text{Var}}\left[\mathit{\widehat \Psi}^{1}(c,a)\right], $$ where $\widehat {\text{Var}}\left[ \mathit{ \widehat \Psi}^{1}(c,a)\right]$ is the estimated variance of the influence curve. A $(1- \alpha)\%$ confidence interval can be obtained as $(\widehat \psi(c,a) \pm z_{1-\alpha/2} \times \widehat \sigma_{\widehat \psi(c,a)})$.

We can obtain the influence curve-based estimator of the standard error for the average treatment effect for $\widehat \delta_{\psi}(c,a,a^\prime)$ by estimating the variance of the sampling distribution of $\widehat \delta_{\psi}(c,a,a^\prime)$ as $$\widehat \sigma^{2}_{\widehat \delta_{\psi}(c,a,a^\prime)}= \frac{1}{n}\widehat {\text{Var}}\left[\widehat \Delta^{1}_{\widehat \psi}(c,a, a^\prime)\right].$$ We can obtain a $(1- \alpha)\%$ confidence interval as $(\widehat \delta_{\psi}(c,a,a^\prime) \pm z_{1-\alpha/2} \times \widehat \sigma_{\widehat \delta_{\psi}(c,a,a^\prime)})$.

\clearpage
\setcounter{figure}{0}
\setcounter{equation}{0}
\section{Estimation compared to ``transportability analyses''}\label{appendix:connection_transportability}
As shown in \ref{appendix:connection_transportability1}, under conditions A1 through A3, and A5, there are three ways to identify $\E[Y^a|C=c]$, using $\phi(c, a)$, $\phi(c,a)$ or $\chi(c,a)$. In this section, we describe a doubly robust estimator $\widehat \chi(c,a)$ for $\chi(c,a)$ that is closely related to estimators for transportability analyses in prior work by our group and others. We also compare aspects of the asymptotic behavior of $\widehat \chi(c,a)$ versus $\widehat \phi(c,a)$, when all necessary models are correctly specified.

To begin, recall that under conditions A1 through A3, and A5, $\E [Y^a|C=c ]$ is identified by $\chi(c, a) \equiv \E \big[ \E[ Y |X, R_{(c)}=1, A=a] |C=c  \big]$. 

\paragraph{Influence function for transportability analyses:} The influence function for $\chi(c, a)$ under the nonparametric model $\mathcal M_{\text{\tiny np}}$ is
\begin{equation*}
    \begin{split}
\mathcal{X}^1_{p_0}(c, a) = \dfrac{1}{\Pr_{p_0}[C=c]} \Bigg\{ \dfrac{I( R_{(c)} = 1 , A = a) \Pr_{p_0}[C = c | X] }{\Pr_{p_0}[R_{(c)} = 1| X]\Pr_{p_0}[A = a| X, R_{(c)} = 1]} \big\{ Y - \E_{p_0}[Y | X, R_{(c)} = 1, A = a]  \big\} \\ + I(C=c)\big\{ \E_{p_0}[Y|X, R_{(c)} = 1, A = a] - \chi_{p_0}(c, a) \big\} \Bigg\}.
    \end{split}
\end{equation*}

Similar to the arguments used for $\mathit\Phi^1_{p_0}(c, a)$, we can conclude that $\mathcal{X}^1_{p_0}(c, a)$, the unique (efficient) influence function for $\chi(c,a)$ under the nonparametric model, is also the efficient influence function under the semiparametric models $\mathcal M_{\text{\tiny semi}}^{*}$ and $\mathcal M_{\text{\tiny semi}}^{**}$.

\paragraph{Efficiency:}
The variance bound for $\chi(c, a)$ is
\begin{equation}\label{eq:var_bound_chi}
  \begin{split}
  &\E\left[ \big( \mathcal{X}^1_{p_0}(c, a) \big)^2 \right] = \\
  &\quad \frac{1}{(\Pr_{p_0}[C=c])^2}  \Bigg\{ \E \left[ \dfrac{I(R_{(c)}=1, A = a) \big(\Pr_{p_0}[ C = c | X]\big)^2}{\big(\Pr_{p_0}[R_{(c)}=1| X] \Pr_{p_0}[A = a | X, R_{(c)}=1]\big)^2}  \big\{ Y - \E_{p_0}[Y | X, R_{(c)}=1, A = a] \big\}^2 \right] \\
  &\quad\quad\quad\quad\quad\quad\quad\quad\quad\quad\quad\quad + \E \left[ I(C = c) \big\{ \E_{p_0}[Y | X, R_{(c)}=1, A = a] - \chi_{p_0}(c, a) \big\}^2 \right] \Bigg \}
  \end{split}
\end{equation}

We will now compare the bounds in expressions (\ref{eq:var_bound_phi}) and (\ref{eq:var_bound_psi}) against the bound in expression (\ref{eq:var_bound_chi}) when conditions A1 through A5 hold, such that $\phi(c, a) = \phi(c,a)= \chi(c,a)$ and any of these functionals can be used to identify $\E[Y^a | C = c]$. 

Under conditions A1 through A5, because $\phi(c, a) = \phi(c,a)= \chi(c,a)$ and $\E[Y | X, C = c, A = a]=\E[Y | X, R_{(c)}=1, A = a]=\E[Y | X, A = a]$, the second term inside the large brackets in each of the expressions (\ref{eq:var_bound_phi}),  (\ref{eq:var_bound_psi}), and (\ref{eq:var_bound_chi}) is the same, so we only need to compare the first term inside the large brackets in these expressions. 

An iterated expectation argument shows that the first such term, in expression  (\ref{eq:var_bound_chi}) is
\begin{equation}
    \begin{split}    
&\E\left[ \dfrac{\E\big[ \big\{Y - \E[Y | X, R_{(c)}=1, A = a] \big\}^2 | X, R_{(c)}=1, A = a\big] \big( \Pr[C = c | X]  \big)^2}{ \Pr[R_{(c)}=1 | X] \Pr[A = a| X, R_{(c)}=1] } \right] \\
&\quad\quad\quad\quad\quad= \E\left[ \dfrac{\mbox{Var}[Y | X, R_{(c)}=1, A = a] \big( \Pr[C = c | X]  \big)^2}{\Pr[R_{(c)}=1 | X]\Pr[A = a| X, R_{(c)}=1] } \right].
    \end{split}
\end{equation}
As shown in the previous subsection, using an iterated expectation argument, the first term in expression (\ref{eq:var_bound_phi}) can be written as (\ref{eq:first_term_var_bound_phi}); likewise the first such term in expression (\ref{eq:var_bound_psi}) can be written as (\ref{eq:first_term_var_bound_psi}).

By condition A4, we have that $$\mbox{Var}[Y | X, R_{(c)}=1, A = a]=\mbox{Var}[Y | X, C = c, A = a]=\mbox{Var}[Y | X, A = a].$$ Furthermore, as noted earlier, by the law of total probability 
\begin{equation*} 
    \begin{split}
    {\Pr}[A = a | X] &= {\Pr}[A =a | X, R_{(c)}=1] {\Pr}[R_{(c)}=1 | X] + {\Pr}[A = a | X, R_{(c)}=0] {\Pr}[ R_{(c)}= 0 | X] \\
  &= {\Pr}[A =a | X, R_{(c)}=1] {\Pr}[R_{(c)}=1 | X] + {\Pr}[A = a | X, C=c] {\Pr}[ C=c| X]. 
    \end{split}
\end{equation*}
Consequently these two facts, and straightforward calculations, establish that $$\E\left[ \big( \mathit{\Psi}^1_{p_0}(c, a) \big)^2  \right] \leq \E\left[ \big( \mathcal{X}^1_{p_0}(c, a) \big)^2  \right], $$
but $\E\left[ \big( \mathit{\Phi}^1_{p_0}(c, a) \big)^2 \right] $ and $\E\left[ \big( \mathcal{X}^1_{p_0}(c, a) \big)^2  \right]$ cannot be ordered without additional assumptions.

\clearpage
\setcounter{figure}{0}
\setcounter{equation}{0}
\section{Additional analyses of the simulation}\label{appendix:simulation}
Here we report additional analyses from the simulation study and additional simulation details. We list the values for the original coefficients for $\beta_k$, for $k=2$ to $k=10$, containing the set of regression coefficients for center level $k$. In the scenario with stronger effect modification and selection, we double the coefficient values for $\beta_{k,3}$.

\vspace{0.15in}
\begin{table}[!ht]
\caption{Coefficients used in the multinomial model for center membership.}
    \centering
\begin{tabular}{lccccccccc}
\toprule
\multirow{2}{*}{\begin{tabular}[c]{@{}l@{}}Element of \\ of $\beta_k$\end{tabular}} & \multicolumn{9}{c}{$k$}                                               \\ \cline{2-10} 
                                                                                  & 2     & 3     & 4     & 5     & 6     & 7     & 8     & 9     & 10    \\ \midrule
$\beta_{k,0}$                                                                     & 0.75  & 1.03  & 0.36  & 0.48  & 0.75  & 0.65  & 0.76  & -0.09 & 1.46  \\
$\beta_{k,1}$                                                                     & -0.36 & -0.18 & -0.32 & -0.13 & -0.47 & -0.42 & -0.52 & -0.4  & -0.19 \\
$\beta_{k,2}$                                                                     & -0.14 & 0.01  & -0.04 & -0.18 & 0.15  & -0.24 & -0.12 & -0.09 & -0.16 \\
$\beta_{k,3}$                                                                     & 0.36  & 0.18  & 0.44  & 0.35  & 0.34  & 0.37  & 0.34  & 0.26  & 0.28  \\ \hline
\end{tabular}
 \caption*{Coefficients used in the multinomial model for center membership, for  $\beta_k= (\beta_{k,0} \ldots, \beta_{k,3})$, for $k=2$ to $k=10$. Here $\beta_{k,0}$ is the intercept in center $k$; $\beta_{k,1}$ is the coefficient of $X_1$ for membership in center $k$; $\beta_{k,2}$ is the coefficient of $X_2$ for membership in center $k$; $\beta_{k,3}$ is the coefficient of $X_3$ for membership in center $k$. }

\end{table}

\subsection{Scenario with stronger effect modification and selection}
In the main text, we summarized the center-specific bias and center-specific MSE for the scenario with stronger effect modification and selection, for all estimators. Here we focus on the proposed methods that explicitly target the center-specific average treatment effect. Appendix Figure \ref{fig:sim_bias_scenario2} shows that all proposed estimators are relatively unbiased and that $\widehat \delta_{\psi(c, 1,0)}$ has the smallest variability around the true value. Appendix Table \ref{Table:sim_coverage_scenario2} shows the coverage, average confidence interval width, and average standard error using the influence curve based estimator of the sampling variance for $\widehat \delta_{\phi}(c, 1,0)$ and $\widehat \delta_{\psi}(c, 1,0)$; or using the standard error estimator for the coefficient of treatment from ordinary least squares regression for $\widehat \delta_{\tau}(c, 1,0)$. For all proposed estimators, coverage of the Wald-style intervals centered around our estimators was close to nominal. As described in the main text, we found that $\widehat \delta_{\psi}(c,1,0)$ always had the smallest average standard error and average confidence width, followed by $\widehat \delta_{\phi}(c,1,0)$; $\widehat \delta_{\tau}(c,1,0)$ had much larger standard error and wider confidence intervals. Appendix Table \ref{Table:sim_SE_scenario2} compares the average of the estimated standard errors to the empirical (Monte Carlo) standard deviation of the estimators over 1000 runs of the simulation. For both estimators, the average of the influence curve based estimated standard errors was almost identical to the average of the estimated standard error using the sandwich; both averages were similar to the estimated standard deviation of the estimator.

\begin{figure}[!ht]
    \caption{Boxplots of deviations from the true center-specific average treatment effect.}
    \centering
    \includegraphics[width=17cm]{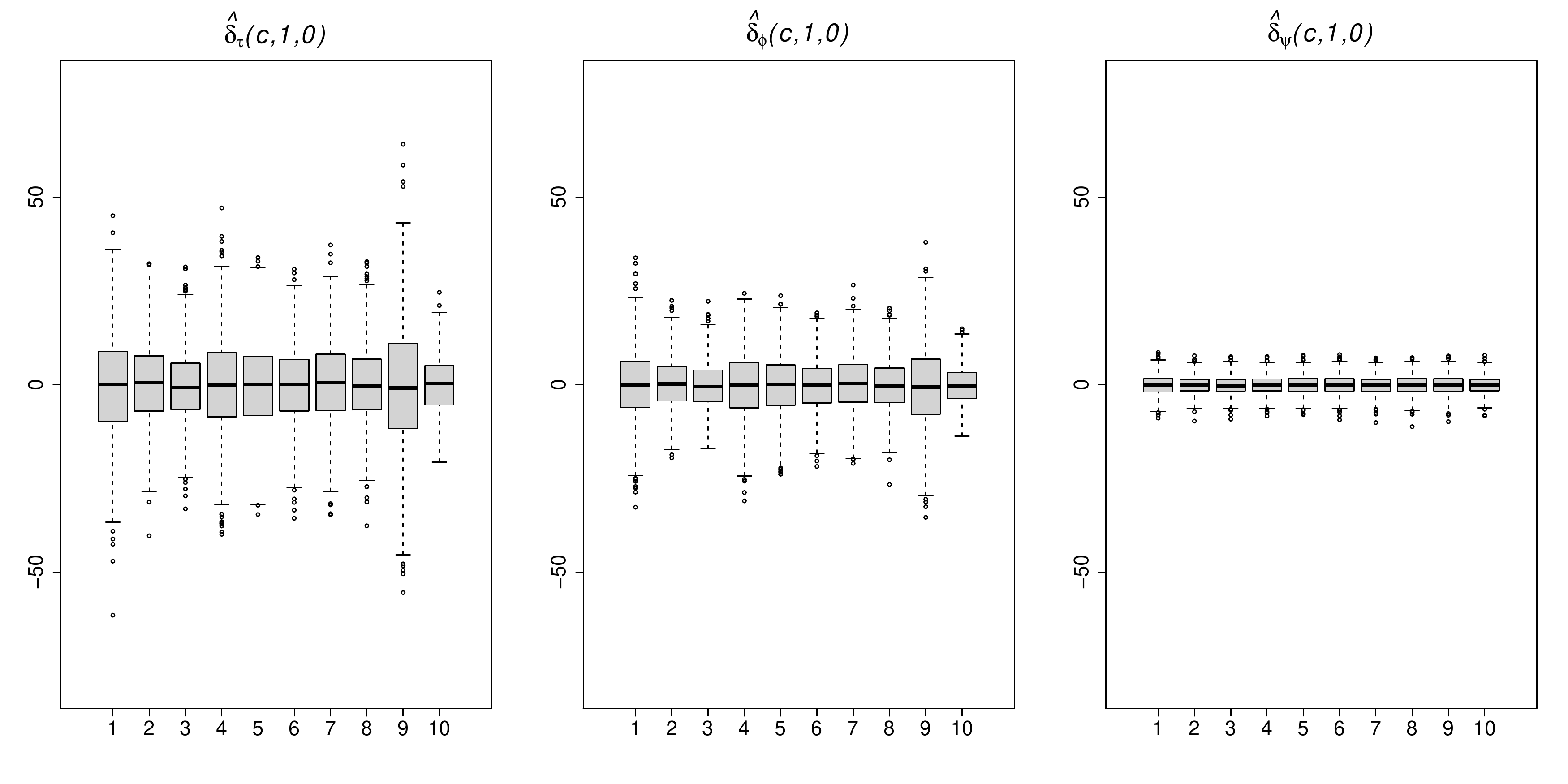}
    \label{fig:sim_bias_scenario2}
\end{figure}

\begin{table}[!ht]
\centering
 \caption{Estimates of coverage, average CI width, and average SE using 95\% Wald-style confidence intervals based on the influence curve in the scenario with stronger effect modification.}
\label{Table:sim_coverage_scenario2}
\resizebox{\textwidth}{!}{\begin{tabular}{@{}ccccccccccc@{}}
\toprule
\\[-7pt]
\multicolumn{1}{c}{\multirow{2}{*}{Center, $c$}} & \multicolumn{1}{c}{\multirow{2}{*}{Average $n_c$}} & \multicolumn{3}{c}{Coverage}                                                                           & \multicolumn{3}{c}{Average CI width}                                                                   & \multicolumn{3}{c}{Average SE}                                                                         \\ \cline{3-11}   \\[-7pt]
\multicolumn{1}{c}{}                             & \multicolumn{1}{l}{}                               & $\widehat \delta_{\tau(c, 1,0)}$ & $\widehat \delta_{\phi(c, 1,0)}$ & $\widehat \delta_{\psi(c, 1,0)}$ & $\widehat \delta_{\tau(c, 1,0)}$ & $\widehat \delta_{\phi(c, 1,0)}$ & $\widehat \delta_{\psi(c, 1,0)}$ & $\widehat \delta_{\tau(c, 1,0)}$ & $\widehat \delta_{\phi(c, 1,0)}$ & $\widehat \delta_{\psi(c, 1,0)}$ \\[5pt] \midrule
1 & 57 & 0.95 & 0.94 & 0.95 & 59.53 & 42.35 & 23.39 & 15.19 & 10.8 & 5.97 \\
2 & 100 & 0.96 & 0.95 & 0.94 & 45.06 & 32.24 & 18.3 & 11.5 & 8.22 & 4.67 \\
3 & 135 & 0.95 & 0.94 & 0.94 & 38.55 & 27.71 & 16.49 & 9.83 & 7.07 & 4.21 \\
4 & 68 & 0.94 & 0.94 & 0.94 & 54.79 & 39.12 & 21.26 & 13.98 & 9.98 & 5.42 \\
5 & 80 & 0.95 & 0.94 & 0.95 & 50.52 & 36.08 & 20.16 & 12.89 & 9.21 & 5.14 \\
6 & 107 & 0.96 & 0.96 & 0.95 & 43.39 & 31.01 & 18.15 & 11.07 & 7.91 & 4.63 \\
7 & 94 & 0.95 & 0.94 & 0.94 & 46.43 & 33.32 & 18.88 & 11.85 & 8.5 & 4.82 \\
8 & 110 & 0.95 & 0.95 & 0.96 & 42.96 & 30.79 & 18.02 & 10.96 & 7.85 & 4.6 \\
9 & 43 & 0.94 & 0.93 & 0.96 & 68.71 & 48.66 & 25.8 & 17.53 & 12.41 & 6.58 \\
10 & 206 & 0.96 & 0.96 & 0.96 & 31.27 & 22.5 & 14.37 & 7.98 & 5.74 & 3.67 \\ \bottomrule \bottomrule
\end{tabular}}
\caption*{Average $n_c$ is the average sample size in center $c$ across simulations. Average CI width = average confidence interval width across simulations; Average SE = standard error from the influence curve (square root of the estimated sampling variance) across simulations.}
\end{table}

\begin{table}[!ht]
\centering
 \caption{Average of the estimated standard errors to the empirical (Monte Carlo) standard deviation in the scenario with stronger effect modification.}
\label{Table:sim_SE_scenario2}
\resizebox{\textwidth}{!}{\begin{tabular}{@{}cccccccc@{}}
\toprule
\\[-7pt]
\multicolumn{1}{c}{\multirow{2}{*}{Center, $c$}} & \multicolumn{1}{c}{\multirow{2}{*}{Average $n_c$}} & \multicolumn{2}{c}{Influence curve}                                                                           & \multicolumn{2}{c}{Sandwich}                                                                   & \multicolumn{2}{c}{Empirical}                                                                         \\ \cline{3-8}   \\[-7pt]
\multicolumn{1}{c}{}                             & \multicolumn{1}{l}{}                             & $\widehat \delta_{\phi(c, 1,0)}$ & $\widehat \delta_{\psi(c, 1,0)}$  & $\widehat \delta_{\phi(c, 1,0)}$ & $\widehat \delta_{\psi(c, 1,0)}$  & $\widehat \delta_{\phi(c, 1,0)}$ & $\widehat \delta_{\psi(c, 1,0)}$ \\[5pt] \midrule
1 & 57 & 10.80 & 5.97 & 10.80 & 5.97 & 10.81 & 6.00 \\
2 & 100 & 8.22 & 4.67 & 8.22 & 4.66 & 8.01 & 4.74 \\
3 & 135 & 7.07 & 4.21 & 7.07 & 4.21 & 7.27 & 4.37 \\
4 & 68 & 9.98 & 5.42 & 9.97 & 5.42 & 10.25 & 5.60 \\
5 & 80 & 9.21 & 5.14 & 9.20 & 5.14 & 9.43 & 5.20 \\
6 & 107 & 7.91 & 4.63 & 7.91 & 4.63 & 7.83 & 4.61 \\
7 & 94 & 8.50 & 4.82 & 8.49 & 4.81 & 8.66 & 4.91 \\
8 & 110 & 7.85 & 4.60 & 7.85 & 4.60 & 7.74 & 4.57 \\
9 & 43 & 12.41 & 6.58 & 12.41 & 6.59 & 12.48 & 6.40 \\
10 & 206 & 5.74 & 3.67 & 5.74 & 3.67 & 5.67 & 3.62 \\ \bottomrule
\end{tabular}}
\caption*{Average $n_c$ is average sample size in center $c$ across simulations. Influence curve = average (over simulation runs) of the estimated standard error using the influence curve (square root of the estimated sampling variance) across simulations; Sandwich = average (over simulation runs) of the standard error using the sandwich estimator of the variance; Empirical = Monte Carlo standard deviation over the estimator across simulation runs.}
\end{table}

\clearpage

\subsection{Scenario with weaker effect modification and selection}

Here, we report the results for the scenario with weaker (baseline) effect modification and selection. In general, the results were similar to the scenario with stronger effect modification and selection. The estimators for center-specific average treatment effects were approximately unbiased and showed the same MSE pattern as in the scenario with stronger effect modification and section (presented in the main text). The pooled, FE1, and FE2 estimators were still biased (for most centers), but they had less bias and lower MSE compared to the scenario with stronger effect modification and selection. Detailed results are presented in Appendix Table \ref{Table:sim_bias_scenario1}, Appendix Table \ref{Table:sim_MSE_scenario1}, and Appendix Figure \ref{fig:sim_bias_scenario1}. Appendix Table \ref{Table:sim_coverage_scenario1} shows that patterns in coverage, average confidence interval width, and average standard error comparing center-specific average treatment effect estimators were similar to those in the scenario with stronger effect modification and selection. Appendix \ref{Table:sim_SE_scenario2} compares the average of the estimated standard errors to the empirical (Monte Carlo) standard deviation of the estimators over 1000 runs of the simulation. Again, the patterns were similar to those in the scenario with stronger effect modification and selection. 

\vspace{0.2in}
\begin{table}[!ht]
\centering
 \caption{Bias results from the simulation for each method.}
\label{Table:sim_bias_scenario1}
\begin{tabular}{@{}cccccccc@{}}
\toprule
Center, $c$ & Average $n_c$ & $\widehat \delta_{\tau}(c,1,0)$ & $\widehat \delta_{\phi}(c,1,0)$ & $\widehat \delta_{\psi}(c,1,0)$  & Pooled & FE1 & FE2  \\ \midrule
1 & 52 & -0.09 & 0.15 & -0.17 & 6.00 & 6.00 & 5.96 \\
2 & 103 & -0.04 & 0.14 & -0.14 & -1.37 & -1.38 & -1.41 \\
3 & 137 & -0.10 & -0.20 & -0.16 & 2.21 & 2.20 & 2.17 \\
4 & 69 & -0.94 & -0.38 & -0.11 & -0.59 & -0.59 & -0.63 \\
5 & 79 & -0.61 & -0.06 & -0.14 & 3.32 & 3.32 & 3.28 \\
6 & 106 & -0.51 & -0.19 & -0.13 & -3.59 & -3.60 & -3.63 \\
7 & 94 & 0.08 & -0.35 & -0.16 & -2.54 & -2.54 & -2.58 \\
8 & 106 & 0.41 & 0.16 & -0.15 & -4.72 & -4.73 & -4.76 \\
9 & 44 & 0.81 & 0.23 & -0.15 & -2.33 & -2.34 & -2.37 \\
10 & 209 & -0.22 & -0.35 & -0.15 & 2.03 & 2.03 & 1.99 \\ \bottomrule
\end{tabular}
\caption*{Average $n_c$ is the average sample size in center $c$ across simulations. See the main text for the definitions of the pooled, FE1, and FE2 estimators.}
\end{table}

\begin{table}[!ht]
\centering
 \caption{MSE results from the simulation for each method.}
\label{Table:sim_MSE_scenario1}
\begin{tabular}{@{}cccccccc@{}}
\toprule
Center, $c$ & Average $n_c$ & $\widehat \delta_{\tau}(c,1,0)$ & $\widehat \delta_{\phi}(c,1,0)$ & $\widehat \delta_{\psi}(c,1,0)$  & Pooled & FE1 & FE2  \\ \midrule
1 & 52 & 298.58 & 117.28 & 14.68 & 52.32 & 52.27 & 41.03 \\
2 & 103 & 147.41 & 52.75 & 9.49 & 18.16 & 18.20 & 7.45 \\
3 & 137 & 114.44 & 39.95 & 8.90 & 21.15 & 21.14 & 10.16 \\
4 & 69 & 241.36 & 90.27 & 11.83 & 16.63 & 16.66 & 5.86 \\
5 & 79 & 206.60 & 66.13 & 9.79 & 27.33 & 27.31 & 16.25 \\
6 & 106 & 158.16 & 57.47 & 9.90 & 29.20 & 29.26 & 18.67 \\
7 & 94 & 186.56 & 58.51 & 9.68 & 22.72 & 22.77 & 12.11 \\
8 & 106 & 144.92 & 54.74 & 10.04 & 38.60 & 38.68 & 28.17 \\
9 & 44 & 385.45 & 138.66 & 16.10 & 21.73 & 21.78 & 11.10 \\
10 & 209 & 80.86 & 28.96 & 6.89 & 20.41 & 20.41 & 9.44 \\ \bottomrule
\end{tabular}
\caption*{Average $n_c$ is the average sample size in center $c$ across simulations. See the main text for the definitions of the pooled, FE1, and FE2 estimators.}
\end{table}

\begin{figure}[!ht]
    \caption{Boxplots of deviations from the true center-specific average treatment effect.}
    \centering
    \includegraphics[width=17cm]{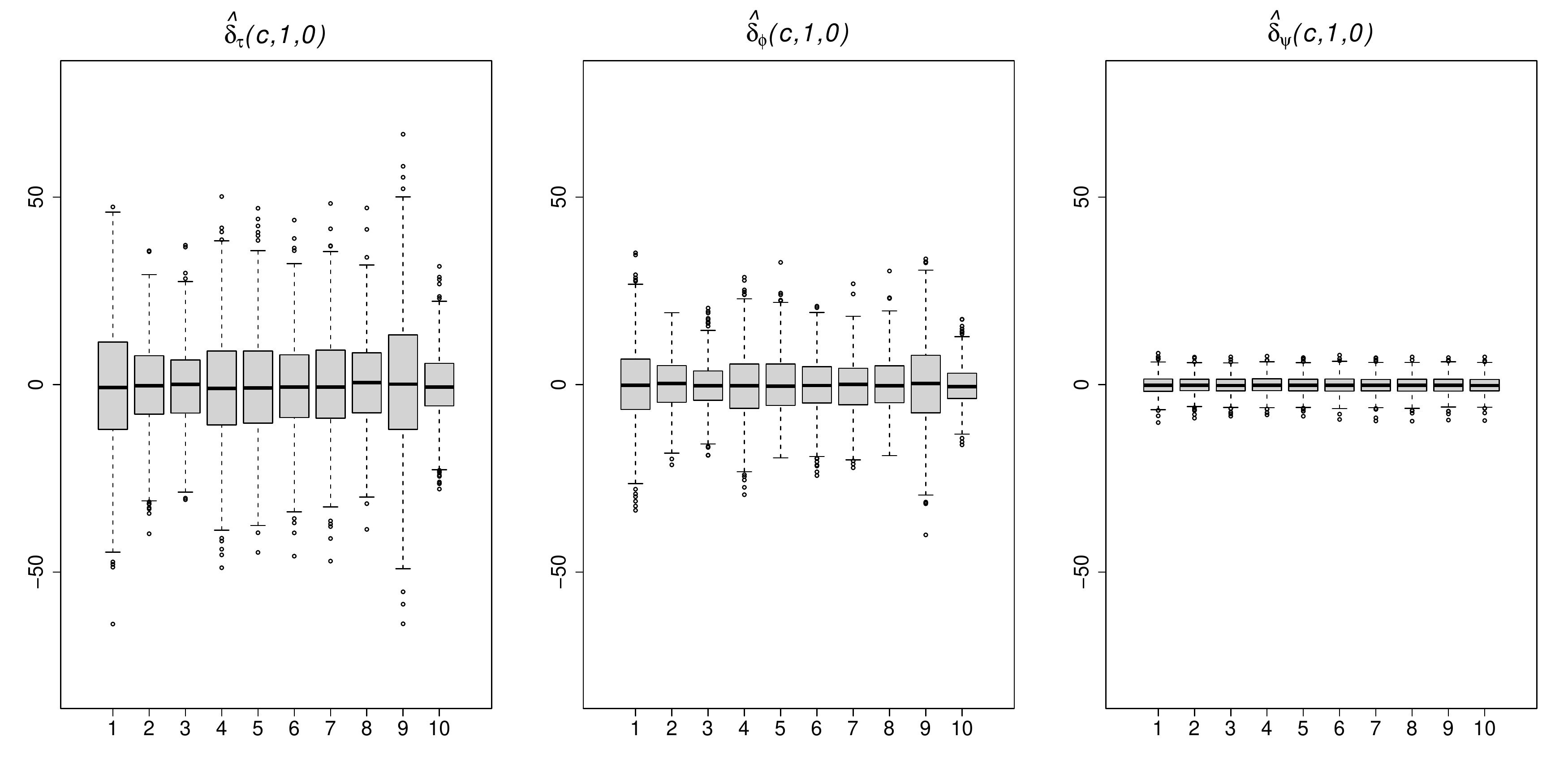}
    \label{fig:sim_bias_scenario1}
\end{figure}

\begin{table}[!ht]
\centering
 \caption{Estimates of coverage, average CI width, and average SE using $95\%$ Wald-style confidence intervals based on the influence curve in the scenario with baseline effect modification.} \label{Table:sim_coverage_scenario1}
\resizebox{\textwidth}{!}{\begin{tabular}{@{}ccccccccccc@{}}
\toprule
\\[-7pt]
\multicolumn{1}{c}{\multirow{2}{*}{Center, $c$}} & \multicolumn{1}{c}{\multirow{2}{*}{Average $n_c$}} & \multicolumn{3}{c}{Coverage}                                                                           & \multicolumn{3}{c}{Average CI width}                                                                   & \multicolumn{3}{c}{Average SE}                                                                         \\ \cline{3-11}   \\[-7pt]
\multicolumn{1}{c}{}                             & \multicolumn{1}{l}{}                               & $\widehat \delta_{\tau(c, 1,0)}$ & $\widehat \delta_{\phi(c, 1,0)}$ & $\widehat \delta_{\psi(c, 1,0)}$ & $\widehat \delta_{\tau(c, 1,0)}$ & $\widehat \delta_{\phi(c, 1,0)}$ & $\widehat \delta_{\psi(c, 1,0)}$ & $\widehat \delta_{\tau(c, 1,0)}$ & $\widehat \delta_{\phi(c, 1,0)}$ & $\widehat \delta_{\psi(c, 1,0)}$ \\[5pt] \midrule
1 & 52 & 0.96 & 0.92 & 0.96 & 69.21 & 40.1 & 15.12 & 17.66 & 10.23 & 3.86 \\
2 & 103 & 0.95 & 0.95 & 0.95 & 49.21 & 29.01 & 12.18 & 12.55 & 7.4 & 3.11 \\
3 & 137 & 0.95 & 0.95 & 0.94 & 42.71 & 25.03 & 11.56 & 10.9 & 6.39 & 2.95 \\
4 & 69 & 0.94 & 0.92 & 0.94 & 59.96 & 34.91 & 13.49 & 15.3 & 8.91 & 3.44 \\
5 & 79 & 0.94 & 0.96 & 0.96 & 55.95 & 32.77 & 13.07 & 14.27 & 8.36 & 3.33 \\
6 & 106 & 0.95 & 0.94 & 0.95 & 48.46 & 28.32 & 12.36 & 12.36 & 7.22 & 3.15 \\
7 & 94 & 0.94 & 0.95 & 0.95 & 51.49 & 30.07 & 12.54 & 13.14 & 7.67 & 3.2 \\
8 & 106 & 0.96 & 0.94 & 0.94 & 48.44 & 28.5 & 12.25 & 12.36 & 7.27 & 3.13 \\
9 & 44 & 0.94 & 0.94 & 0.95 & 75.56 & 43.67 & 15.45 & 19.28 & 11.14 & 3.94 \\
10 & 209 & 0.94 & 0.93 & 0.96 & 34.42 & 20.26 & 10.68 & 8.78 & 5.17 & 2.73 \\ \bottomrule
\end{tabular}}
\caption*{Average $n_c$ is the average sample size in center $c$ across simulations. Average CI width = average confidence interval width across simulation; Average SE = standard error from the influence curve (square root of the estimated sampling variance) across simulations.}
\end{table}

\begin{table}[!ht]
\centering
 \caption{Average of the estimated standard errors to the empirical (Monte Carlo) standard deviation in the scenario with stronger effect modification.}
\label{Table:sim_SE_scenario1}
\resizebox{\textwidth}{!}{\begin{tabular}{@{}cccccccc@{}}
\toprule
\\[-7pt]
\multicolumn{1}{c}{\multirow{2}{*}{Center, $c$}} & \multicolumn{1}{c}{\multirow{2}{*}{Average $n_c$}} & \multicolumn{2}{c}{Influence curve}                                                                           & \multicolumn{2}{c}{Sandwich}                                                                   & \multicolumn{2}{c}{Empirical}                                                                         \\ \cline{3-8}   \\[-7pt]
\multicolumn{1}{c}{}                             & \multicolumn{1}{l}{}                             & $\widehat \delta_{\phi(c, 1,0)}$ & $\widehat \delta_{\psi(c, 1,0)}$  & $\widehat \delta_{\phi(c, 1,0)}$ & $\widehat \delta_{\psi(c, 1,0)}$  & $\widehat \delta_{\phi(c, 1,0)}$ & $\widehat \delta_{\psi(c, 1,0)}$ \\[5pt] \midrule
1 & 52 & 10.23 & 3.86 & 10.23 & 3.86 & 10.83 & 3.83 \\
2 & 103 & 7.40 & 3.11 & 7.40 & 3.11 & 7.26 & 3.08 \\
3 & 137 & 6.39 & 2.95 & 6.38 & 2.95 & 6.32 & 2.98 \\
4 & 69 & 8.91 & 3.44 & 8.90 & 3.44 & 9.49 & 3.44 \\
5 & 79 & 8.36 & 3.33 & 8.36 & 3.33 & 8.13 & 3.13 \\
6 & 106 & 7.22 & 3.15 & 7.22 & 3.16 & 7.58 & 3.14 \\
7 & 94 & 7.67 & 3.20 & 7.66 & 3.20 & 7.64 & 3.11 \\
8 & 106 & 7.27 & 3.13 & 7.26 & 3.13 & 7.40 & 3.16 \\
9 & 44 & 11.14 & 3.94 & 11.14 & 3.94 & 11.77 & 4.01 \\
10 & 209 & 5.17 & 2.73 & 5.17 & 2.72 & 5.37 & 2.62 \\ \bottomrule
\end{tabular}}
\caption*{Average $n_c$ is average sample size in center $c$ across simulations. Influence curve = average (over simulation runs) of the estimated standard error using the influence curve (square root of the estimated sampling variance) across simulations; Sandwich = average (over simulation runs) of the standard error using the sandwich estimator of the variance; Empirical = Monte Carlo standard deviation over the estimator across simulation runs.}
\end{table}

\clearpage
\setcounter{figure}{0}
\setcounter{equation}{0}
\section{Additional results from the HALT-C trial}\label{appendix:haltc}

\paragraph{Stability analysis using generalized additive models:} We used generalized additive models (GAMs) to estimate models for the outcome and participation, using the main effects of all baseline covariates (we always estimated the model for the probability of treatment using parametric logistic regression because that model cannot be misspecified). We used the mgcv package in R \citelatex{wood2001mgcv}, which can automatically find a smoothness penalty, and considered splines for all continuous covariates. When modeling the expectation of the outcome, we used the identity link and all package defaults, which fits penalized thin plate regression splines and allows up to 10 for $k$, the dimension of the basis. When modeling the probability of participation, we used a multinomial link, but had to fit unpenalized thin plate regression splines with the dimension of the basis, $k$ set to 5 because of convergence issues (flexibility of the splines had to be limited due to small center size). We obtained non-parametric bootstrap-based standard errors to form 95\% Wald confidence intervals for the center-specific treatment effects (1,000 resampled datasets). We tuned and refit all working models in each bootstrap sample.


\begin{figure}[!ht]
    \caption{Forest plot from stability analysis using the HALT-C trial data.}
    \centering
    \includegraphics[width=11cm]{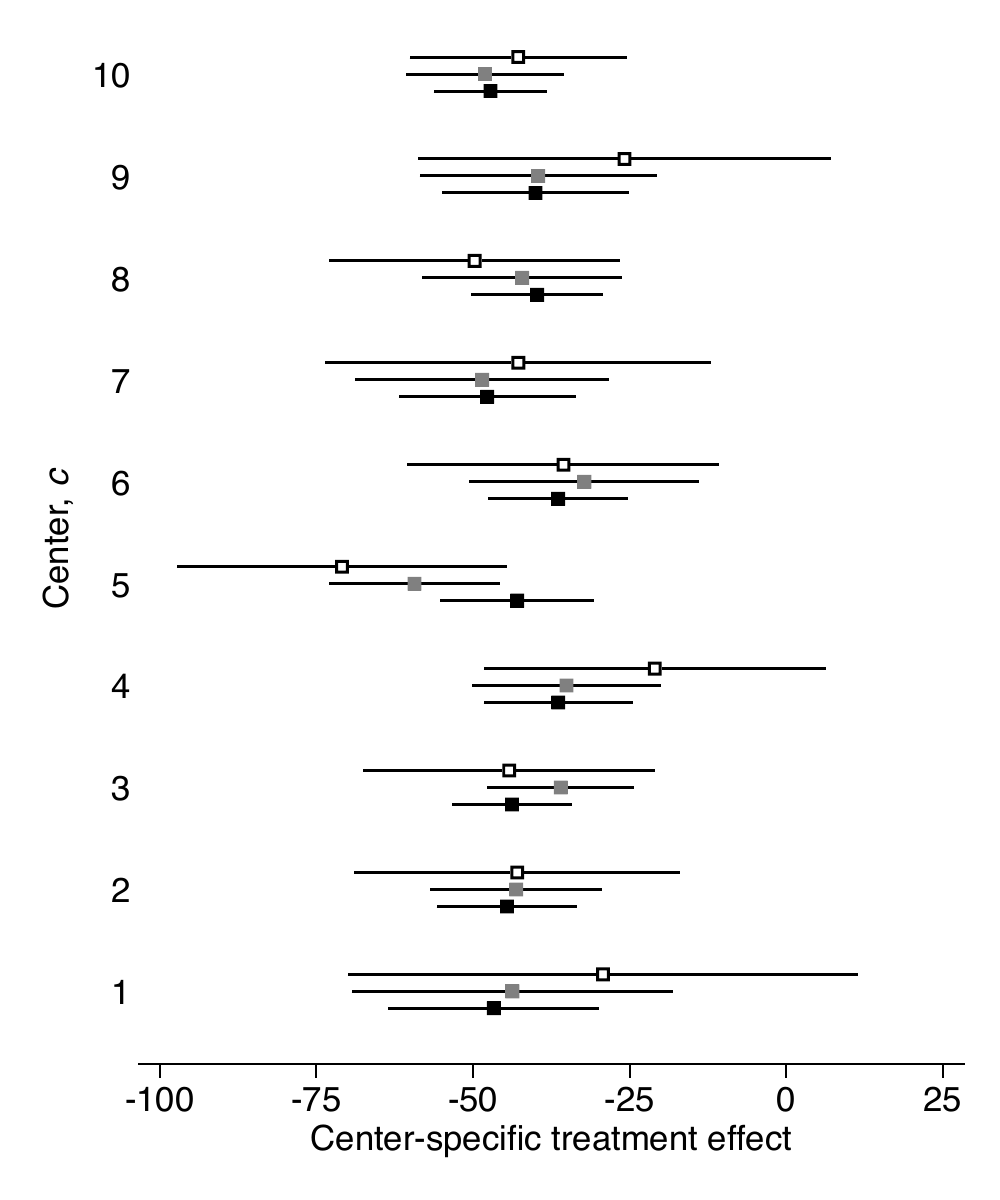}
    \caption*{ATE, center-specific average treatment effects. \\
    Point estimates (square markers) and 95\% confidence intervals (extending lines). White squares represent the crude analysis using the estimator in display \eqref{eq:tau_APP}; grey squares represent the adjusted analysis  when center-outcome associations may be present,  using the estimator in display \eqref{eq:estimatorphi2_APP}; black squares represent the adjusted analysis  when center-outcome associations are absent, using the estimator in display \eqref{eq:estimatorpsi_APP}.}
    \label{fig:forest_plot_appendix}
\end{figure}


\begin{landscape}
\begin{table}[ht]
\centering
 \caption{Baseline characteristics in the HALT-C trial, stratified by center, for centers 1 through 5.} \label{Table:Appendix_centers_set1}
\begin{tabular}{@{}llllll@{}}
\toprule
                  Center, $c$    & 1               & 2               & 3               & 4               & 5               \\ \midrule
Number of individuals                       & 48              & 97              & 130             & 66              & 76              \\
Baseline platelets, $\times$ 1000/mm$^3$   & 175.17 (77.85)  & 162.65 (67.50)  & 168.19 (68.58)  & 167.64 (58.12)  & 178.58 (64.82)  \\
Age in years    & 51.33 (7.04)    & 50.23 (6.56)    & 51.34 (6.40)    & 50.83 (8.35)    & 49.86 (6.57)    \\
Female      & 11 (22.9)       & 29 (29.9)       & 32 (24.6)       & 21 (31.8)       & 23 (30.3)       \\
Received pegylated interferon          & 5 (10.4)        & 25 (25.8)       & 33 (25.4)       & 15 (22.7)       & 19 (25.0)       \\
White           & 28 (58.3)       & 75 (77.3)       & 97 (74.6)       & 47 (71.2)       & 67 (88.2)       \\
Baseline white blood cell count, $\times$ 1000/mm$^3$      & 5.47 (1.78)     & 5.80 (1.92)     & 5.61 (1.81)     & 5.98 (2.08)     & 6.02 (1.57)     \\
Used recreational drugs     & 25 (52.1)       & 46 (47.4)       & 66 (50.8)       & 31 (47.0)       & 33 (43.4)       \\
Received a transfusion      & 16 (33.3)       & 41 (42.3)       & 52 (40.0)       & 24 (36.4)       & 22 (28.9)       \\
Body mass index, weight (kg)/height(m)$^2$        & 30.13 (6.20)    & 29.89 (6.19)    & 29.25 (4.66)    & 29.98 (6.09)    & 30.18 (5.38)    \\
Creatinine, mg/dl     & 0.80 (0.19)     & 0.81 (0.13)     & 0.81 (0.15)     & 0.81 (0.16)     & 0.84 (0.17)     \\
Ever smoked      & 37 (77.1)       & 73 (75.3)       & 100 (76.9)      & 48 (72.7)       & 60 (78.9)       \\
Received interferon and ribavirin    & 44 (91.7)       & 81 (83.5)       & 111 (85.4)      & 45 (68.2)       & 64 (84.2)       \\
Reported diabetes        & 10 (20.8)       & 25 (25.8)       & 13 (10.0)       & 8 (12.1)        & 11 (14.5)       \\
Serum ferritin, ng/ml   & 429.42 (432.51) & 303.35 (317.08) & 436.21 (518.96) & 294.06 (307.87) & 289.01 (251.52) \\
Ultrasound evidence of splenomegaly        & 12 (25.0)       & 30 (30.9)       & 42 (32.3)       & 27 (40.9)       & 28 (36.8)       \\
Ever drank alcohol     & 32 (66.7)       & 85 (87.6)       & 121 (93.1)      & 51 (77.3)       & 66 (86.8)       \\
Hemoglobin, g/dl   & 14.91 (1.31)    & 14.80 (1.18)    & 15.64 (1.46)    & 14.93 (1.43)    & 14.88 (1.38)    \\
Aspartate aminotransferase, U/L         & 92.94 (58.34)   & 96.32 (74.43)   & 85.11 (70.92)   & 97.33 (66.24)   & 81.11 (72.95)   \\ \bottomrule
\end{tabular}
\caption*{Results reported as mean (standard deviation) for continuous variables and count (percentage) for binary variables.\\
$A$, indicates randomization to treatment with peginterferon alfa-2a ($A=1$) versus no treatment ($A=0$); kg, kilogram; m, meter; mg, milligram; dl, deciliter; ml, milliliter; g, gram; U/L, units per liter.}
\end{table}
\end{landscape}

\begin{landscape}
\begin{table}[ht]
\centering
 \caption{Baseline characteristics in the HALT-C trial, stratified by center, for centers 6 through 10.} \label{Table:Appendix_centers_set2}
\begin{tabular}{@{}llllll@{}}
\toprule
                  Center, $c$        & 6               & 7               & 8               & 9               & 10              \\ \midrule
Number of individuals                       & 101             & 89              & 100             & 42              & 199             \\
Baseline platelets, $\times$ 1000/mm$^3$    & 154.53 (60.46)  & 160.08 (68.25)  & 152.68 (52.83)  & 157.31 (70.46)  & 171.30 (66.10)  \\
Age in years    & 52.42 (8.35)    & 49.54 (6.73)    & 50.44 (7.38)    & 50.64 (7.54)    & 50.06 (7.53)    \\
Female      & 30 (29.7)       & 24 (27.0)       & 32 (32.0)       & 9 (21.4)        & 60 (30.2)       \\
Received pegylated interferon       & 35 (34.7)       & 19 (21.3)       & 40 (40.0)       & 15 (35.7)       & 66 (33.2)       \\
White          & 64 (63.4)       & 75 (84.3)       & 72 (72.0)       & 35 (83.3)       & 115 (57.8)      \\
Baseline white blood cell count, $\times$ 1000/mm$^3$       & 5.61 (1.79)     & 5.79 (1.98)     & 5.63 (1.81)     & 5.58 (1.77)     & 5.81 (1.92)     \\
Used recreational drugs     & 41 (40.6)       & 38 (42.7)       & 41 (41.0)       & 17 (40.5)       & 91 (45.7)       \\
Received a transfusion      & 41 (40.6)       & 42 (47.2)       & 44 (44.0)       & 16 (38.1)       & 77 (38.7)       \\
Body mass index, weight (kg)/height(m)$^2$       & 29.65 (5.85)    & 30.67 (5.12)    & 29.88 (5.55)    & 30.44 (5.29)    & 29.93 (5.34)    \\
Creatinine, mg/dl     & 0.91 (0.18)     & 0.92 (0.16)     & 0.83 (0.15)     & 0.90 (0.17)     & 0.84 (0.17)     \\
Ever smoked       & 74 (73.3)       & 68 (76.4)       & 69 (69.0)       & 29 (69.0)       & 155 (77.9)      \\
Received interferon and ribavirin     & 83 (82.2)       & 80 (89.9)       & 85 (85.0)       & 34 (81.0)       & 151 (75.9)      \\
Reported diabetes         & 15 (14.9)       & 13 (14.6)       & 16 (16.0)       & 5 (11.9)        & 50 (25.1)       \\
Serum ferritin, ng/ml  & 440.21 (671.90) & 352.33 (307.69) & 398.18 (399.88) & 486.74 (498.72) & 360.18 (378.14) \\
Ultrasound evidence of splenomegaly         & 31 (30.7)       & 26 (29.2)       & 33 (33.0)       & 13 (31.0)       & 77 (38.7)       \\
Ever drank alcohol     & 80 (79.2)       & 62 (69.7)       & 78 (78.0)       & 36 (85.7)       & 173 (86.9)      \\
Hemoglobin, g/dl   & 14.63 (1.54)    & 14.97 (1.51)    & 14.85 (1.28)    & 14.93 (1.44)    & 14.97 (1.42)    \\
Aspartate aminotransferase, U/L        & 85.70 (46.81)   & 80.64 (51.68)   & 76.92 (44.07)   & 97.38 (57.20)   & 92.72 (48.81)   \\ \bottomrule
\end{tabular}
\caption*{Results reported as mean (standard deviation) for continuous variables and count (percentage) for binary variables.\\
$A$, indicates randomization to treatment with peginterferon alfa-2a ($A=1$) versus no treatment ($A=0$); kg, kilogram; m, meter; mg, milligram; dl, deciliter; ml, milliliter; g, gram; U/L, units per liter.}
\end{table}
\end{landscape}

\clearpage
\begin{table}[!ht]
\centering
 \caption{Analysis of the HALT-C trial, stability analysis using GAMs.} \label{Table:ATE_appendix}
\begin{tabular}{ccp{.2\textwidth}p{.2\textwidth}p{.2\textwidth}}
\toprule
Center, $c$ & $n$ & \hfil$\widehat \delta_{\tau}(c,1,0)$ & \hfil$\widehat \delta_{\phi}(c,1,0)$ & \hfil$\widehat \delta_{\psi}(c,1,0)$  \\ \midrule
1 & 48 & \hfil -29.3 (11.5, -70) \par \hfil {[}20.8{]} & \hfil -43.7 (-18.1, -69.4) \par \hfil {[}13.1{]} & \hfil  -46.7 (-29.8, -63.5) \par \hfil {[}8.6{]} \\
2 & 97 & \hfil  -42.9 (-16.9, -69) \par \hfil {[}13.3{]} & \hfil -43.1 (-29.4, -56.8) \par \hfil {[}7{]} & \hfil -44.6 (-33.4, -55.8) \par \hfil {[}5.7{]} \\
3 & 130 & \hfil  -44.2 (-20.9, -67.5) \par \hfil {[}11.9{]} & \hfil  -36 (-24.3, -47.7) \par \hfil {[}6{]} & \hfil  -43.8 (-34.2, -53.3) \par \hfil {[}4.9{]} \\
4 & 66 & \hfil -21 (6.3, -48.3) \par \hfil {[}13.9{]} & \hfil -35.1 (-20, -50.2) \par \hfil {[}7.7{]} & \hfil  -36.4 (-24.5, -48.3) \par \hfil {[}6.1{]} \\
5 & 76 & \hfil -70.9 (-44.6, -97.3) \par \hfil {[}13.4{]} & \hfil  -59.3 (-45.7, -73) \par \hfil {[}7{]} & \hfil  -43 (-30.6, -55.3) \par \hfil {[}6.3{]} \\
6 & 101 & \hfil  -35.6 (-10.7, -60.5) \par \hfil {[}12.7{]} & \hfil  -32.2 (-13.9, -50.6) \par \hfil {[}9.4{]} & \hfil  -36.4 (-25.3, -47.6) \par \hfil {[}5.7{]} \\
7 & 89 & \hfil  -42.8 (-12, -73.6) \par \hfil {[}15.7{]} & \hfil -48.6 (-28.2, -68.9) \par \hfil {[}10.4{]} & \hfil -47.7 (-33.6, -61.8) \par \hfil {[}7.2{]} \\
8 & 100 & \hfil -49.7 (-26.6, -72.9) \par \hfil {[}11.8{]} & \hfil  -42.1 (-26.2, -58.1) \par \hfil {[}8.1{]} & \hfil  -39.8 (-29.2, -50.3) \par \hfil {[}5.4{]} \\
9 & 42 & \hfil  -25.8 (7.1, -58.8) \par \hfil {[}16.8{]} & \hfil -39.6 (-20.7, -58.5) \par \hfil {[}9.6{]} & \hfil  -40 (-25.2, -54.9) \par \hfil {[}7.6{]} \\
10 & 199 & \hfil  -42.8 (-25.4, -60.1) \par \hfil {[}8.9{]} & \hfil  -48.1 (-35.5, -60.7) \par \hfil {[}6.4{]} & \hfil  -47.2 (-38.2, -56.3) \par \hfil {[}4.6{]} \\ \bottomrule
\end{tabular}
\caption*{The number of individuals in center $c$ is $n_c$. The 95\% Wald confidence intervals are in parentheses. The standard error calculated from the sandwich estimator are in brackets. }
\end{table}

\clearpage
\setcounter{figure}{0}
\setcounter{equation}{0}
\section{Code \& data}\label{appendix:code}

\paragraph{Code for empirical analyses:} We provide \texttt{R} code to reproduce the analyses in Section 7 of the paper. Specifically, we provide the following files:

\begin{itemize}
\item[]   \texttt{01\_run\_haltc\_analysis\_forestplot\_tablecenters.R} produces estimates for the crude and covariate-adjusted estimators in the paper, used to create Table 5 in the main text. The code relies on two source files, \texttt{0\_sourcecode\_no\_transport.R} (crude analysis run separately in each center) and \texttt{0\_sourcecode\_reinterpret\_multicenter.R} (new estimators proposed in the paper). We use the \texttt{R} package \texttt{geex} \citelatex{saul2017} to solve the estimating equations corresponding to estimators in our paper.

\item[]  \texttt{02\_test\_conditional.R} assesses whether center-outcome associations are present, using ANCOVA to compare the expectation of the outcome in a linear regression model that includes the main effects of baseline covariates and treatment and all possible baseline covariates and treatment products, against a linear regression model that additionally includes the main effect of the center indicators and all product terms between baseline covariates, treatment, and center indicators. 

\item[]  \texttt{03\_test\_marginals.R} assesses homogeneity of average treatment effects across centers. It uses $\widehat \phi(c,a)$ and $\widehat \psi(c,a)$ to estimate the components of the contrasts involved in the null hypothesis of homogeneity, then uses an omnibus Wald chi-square test for assessing whether the center-specific treatment effects are homogeneous across centers. 

\end{itemize}

\paragraph{Data availability:} The HALT-C trial data are not publicly available, but they can be obtained from the National Institute of Diabetes and Digestive and Kidney Diseases (NIDDK) Central Repository (\url{https://repository.niddk.nih.gov/studies/halt-c/}; last accessed April 12, 2021).

\clearpage
\renewcommand{\thesection}{Appendix References}

\section{}\label{appendix:references}
\bibliographystylelatex{unsrt}
\bibliographylatex{ref_multicenter}